\documentclass[12pt]{article}
\pdfoutput=1

\usepackage{amsthm}
\usepackage{putex}
\usepackage{pgf,interval}
\usepackage{graphicx}
\usepackage{caption}
\usepackage{amsmath}
\usepackage{array}
\usepackage{subcaption}
\usepackage{epstopdf}
\usepackage{enumerate}
\usepackage{cite}
\usepackage{youngtab}
\usepackage{tensor}
\usepackage{slashed}
\usepackage[aligntableaux=center]{ytableau}
\usepackage[utf8]{inputenc}
\usepackage{rotating}
\usepackage{bigfoot}
\usepackage[
      colorlinks=true,
      linkcolor=blue,
      urlcolor=blue,
      filecolor=black,
      citecolor=red,
      ]{hyperref}

      \usepackage{tikz,varwidth}
      \usetikzlibrary{decorations.pathreplacing,positioning}
      \usetikzlibrary{calc}

\newcommand{\abs}[1]{\left\lvert #1 \right\rvert}

\newcommand {\be} {\begin {equation}}
\newcommand {\ee} {\end {equation}}

\newcommand {\bes} {\begin {equation*}}
\newcommand {\ees} {\end {equation*}}

\newcommand{\es}[2] {\begin{equation} \label{#1} \begin{split} #2 \end{split} \end{equation}}

\newcommand{\CP}{\mathbb{CP}}
\newcommand{\Z}{\mathbb{Z}}

\newcommand{\cA}{{\mathcal A}}
\newcommand{\cB}{{\mathcal B}}

\newcommand{\cE}{{\mathcal E}}
\newcommand{\cF}{{\mathcal F}}

\newcommand{\cN}{{\mathcal N}}
\newcommand{\cO}{{\mathcal O}}

\newcommand{\cS}{{\mathcal S}}
\newcommand{\cT}{{\mathcal T}}

\newcommand{\cM}{{\mathcal M}}

\newcommand{\beq}{\begin{equation}}
\newcommand{\eeq}{\end{equation}}

\newcommand\ep{\epsilon}

\def\ie{\begin{equation}\begin{aligned}}
\def\fe{\end{aligned}\end{equation}}

\newcommand{\m}{\mu}

\newcommand{\A}{{\alpha}}
\newcommand{\B}{{\beta}}
\newcommand{\D}{{\delta}}

\newcommand{\mZ}{{\mathbb Z}}
\newcommand{\mR}{{\mathbb R}}

\newcommand{\Mands}{{\bf s}}
\newcommand{\Mandt}{{\bf t}}
\newcommand{\Mandu}{{\bf u}}

\numberwithin{equation}{section}

\def\<{\langle}
\def\>{\rangle}

\def\half{{\scriptstyle \frac 12}}

\def\threeh{{\scriptstyle \frac 32}}
\def\fiveh{{\scriptstyle \frac 52}}

\newtheorem{theorem}{Theorem}[section]
\newtheorem{corollary}{Corollary}[theorem]
\newtheorem{lemma}[theorem]{Lemma}

\begin{document}

\preprint{PUPT-2607\\ QMUL-PH-19-37}

\institution{Weizmann}{Department of Particle Physics and Astrophysics, Weizmann Institute of Science, Rehovot, Israel}
\institution{QueenMary}{School of Physics and Astronomy, Queen Mary University of London, London, E1 4NS, UK}
\institution{DAMTP}{DAMTP, Wilberforce Road, Cambridge CB3 0WA, UK}
\institution{PU}{Joseph Henry Laboratories, Princeton University, Princeton, NJ 08544, USA}
\institution{HarvardCMSA}{Center of Mathematical Sciences and Applications and Jefferson Physical Laboratory, \cr Harvard University, Cambridge, MA 02138, USA}

\title{
Modular Invariance in Superstring Theory From $\cN = 4$ Super-Yang-Mills
}

\authors{Shai M.~Chester,\worksat{\Weizmann} Michael B.~Green,\worksat{\QueenMary, \DAMTP} Silviu S.~Pufu,\worksat{\PU} Yifan Wang,\worksat{\PU,\HarvardCMSA}\\[10 pt] and Congkao Wen\worksat{\QueenMary}}

\abstract{
We study the four-point  function of the lowest-lying half-BPS operators in the ${\cal N} =4$ $SU(N)$ super-Yang-Mills  theory and its relation to the flat-space four-graviton amplitude in type IIB superstring theory.  We work in a  large-$N$ expansion in which the complexified Yang-Mills coupling $\tau$ is fixed.  In this expansion, non-perturbative instanton contributions are present, and the $SL(2, \mathbb{Z})$ duality invariance of correlation functions is manifest.   Our results are based on  a detailed analysis of the sphere partition function of the mass-deformed SYM theory, which was previously computed using supersymmetric localization.  This partition function determines a certain integrated correlator in the undeformed ${\cal N} = 4$ SYM theory, which in turn constrains the four-point correlator at separated points. In a normalization where the two-point functions are proportional to $N^2-1$ and are independent of $\tau$ and $\bar \tau$, we find that the terms of order $\sqrt{N}$ and $1/\sqrt{N}$ in the large $N$ expansion of the four-point correlator are proportional to the non-holomorphic Eisenstein series $E({\scriptstyle \frac{3}{2}},\tau,\bar\tau)$ and $E({\scriptstyle \frac{5}{2}},\tau,\bar\tau)$, respectively. In the flat space limit, these terms match the corresponding terms in the type IIB S-matrix arising from $R^4$ and $D^4 R^4$ contact interactions, which, for the $R^4$ case, represents a check of AdS/CFT at finite string coupling.  Furthermore, we present striking evidence that these results generalize so that, at order $N^{\frac{1}{2}-m}$ with integer $m \ge 0$, the expansion of the integrated correlator we study is a linear sum of non-holomorphic Eisenstein series with half-integer index, which are manifestly  $SL(2,\mathbb{Z})$ invariant.  }

\date{}

\maketitle

\tableofcontents

\pagebreak

\section{Introduction}
\label{introduction}

\subsection{Overview}
\label{overview}

It has long been appreciated that string theory can be given a non-perturbative definition through the anti-de Sitter / conformal field theory (AdS/CFT) correspondence \cite{Maldacena:1997re,Witten:1998qj,Gubser:1998bc}.  In principle, such a definition allows one to calculate properties of closed string theory or M-theory, and hence of quantum gravity, in terms of properties of certain CFTs such as the $SU(N)$ ${\cal N} = 4$ supersymmetric Yang-Mills (SYM) theory, without any explicit reference to standard stringy methods.  
For instance, as originally proposed in \cite{Polchinski:1999ry,Susskind:1998vk,Giddings:1999jq} and further explored in \cite{Penedones:2010ue,Fitzpatrick:2011hu,Fitzpatrick:2011jn}, one can obtain flat space scattering amplitudes of massless particles by taking certain kinematic limits of CFT correlation functions.  While in general it is very difficult to determine the CFT correlation functions in the regime of interest, recent progress combining various analytic bootstrap techniques \cite{Alday:2013opa, Rastelli:2016nze, Aharony:2016dwx, Rastelli:2017ymc,Rastelli:2017udc, Caron-Huot:2018kta,Alday:2014qfa,Alday:2018pdi,Alday:2017xua, Aprile:2017bgs, Aprile:2017qoy, Alday:2017vkk, Aprile:2018efk, Aprile:2019rep, Alday:2019nin, Drummond:2019hel, Giusto:2018ovt, Rastelli:2019gtj, Giusto:2019pxc} and supersymmetry has allowed the determination of CFT correlators, at strong coupling, in a variety of examples arising from top-down AdS/CFT constructions \cite{Chester:2018aca,Chester:2018dga,Binder:2018yvd,Binder:2019mpb,Alday:2018pdi,Binder:2019jwn,Chester:2019pvm}.  From the CFT correlators in these examples, one can then extract the first few terms in the low-energy expansion of scattering amplitudes in superstring theory or M-theory.   In the cases studied so far, these scattering amplitudes had been known from other string theory arguments, and therefore these calculations represent remarkable precision tests of AdS/CFT\@.  

The main focus of this work is on the ${\cal N} = 4$ SYM theory, and, in particular, on how one can reproduce the $SL(2, \Z)$ modular properties of superstring amplitudes from CFT correlators.  Before discussing these modular properties, however, let us briefly review the recent work on the four-point correlators of ${\cal N} = 4$ SYM in the 't Hooft limit, and then discuss how these computations should be modified in order to access $SL(2, \Z)$ invariants.  As we will explain, in order to uncover the modular properties, we should consider a strong coupling limit that is different from the usual 't Hooft limit.

\subsubsection*{CFT correlation functions in 't Hooft limit and type IIB graviton amplitudes}

Refs.~\cite{Alday:2018pdi,Binder:2019jwn,Chester:2019pvm} studied the four-point function of superconformal primary operators in the same multiplet as the stress-energy tensor of the ${\cal N} = 4$ SYM theory, considered in the large-$N$ 't Hooft limit, where the 't Hooft coupling $\lambda \equiv g_\text{YM}^2 N$ is fixed and in the regime $\lambda \gg 1$.  In a normalization in which the disconnected term in the four-point function scales as $N^4 \lambda^0$ at large $N$ and $\lambda$, the leading connected contribution is of order $N^{2} \,\lambda^0$ and reproduces the tree amplitude of type IIB supergravity.   The next correction, of order $N^2\, \lambda^{-\threeh}$, reproduces the ``stringy'' correction corresponding  to the eight-derivative contact interaction of the form $R^4$  (a well-known contraction of four Riemann tensors) in the effective IIB superstring action \cite{Gross:1986iv,Grisaru:1986dk}.  Several other terms of higher order in $1/\lambda$ and in $1/N$ have been matched with terms that arise in the low energy expansion of string perturbation theory \cite{Binder:2019jwn,Chester:2019pvm}.
    
The procedure that allows the above comparisons between string perturbation theory and ${\cal N} = 4$ correlation functions is as follows.  For the four-point function mentioned above, the analytic bootstrap consistency conditions \cite{Rastelli:2016nze,Binder:2019jwn,Chester:2019pvm} determine the position dependence of the tree-level Witten diagrams corresponding to higher derivative terms in the string effective action, up to a number of undetermined coefficients.  At low orders in the derivative expansion, these coefficients can be fully determined, in principle, from constraints derived using supersymmetric localization \cite{Chester:2018aca,Chester:2018dga,Binder:2018yvd,Binder:2019jwn,Chester:2019pvm,Binder:2019mpb}.  Indeed, as shown in \cite{Binder:2019jwn} in the case of the ${\cal N} = 4$ SYM theory, one can obtain integrated constraints on the 4-point functions of the operators in the stress-tensor multiplet by taking four derivatives of the partition function $Z$ of the mass-deformed ${\cal N} = 4$ SYM theory placed on a round $S^4$.  This mass deformation of the ${\cal N} = 4$ SYM theory is  referred to as the ${\cal N} = 2^*$ theory, and its $S^4$ partition function was computed by Pestun using supersymmetric localization \cite{Pestun:2007rz}.  As shown in \cite{Pestun:2007rz}, this partition function takes the form of a finite-dimensional integral over constant values of one of the vector multiplet scalars.  The integrand is a product of a classical contribution, a one-loop contribution, and contributions from instantons located at the north and south poles of $S^4$ \cite{Moore:1997dj,Losev:1997tp,Nekrasov:2002qd,Nekrasov:2003rj}.  

The procedure  implemented in  \cite{Binder:2019jwn}  for analyzing the $\cN=4$ integrated correlation function made use of the constraint coming from the $m=0$ limit of the  mixed derivative (where $m$ is the mass deformation parameter)  
 \es{MixedDer}{
   \tau_2^2 \partial_\tau \partial_{\bar \tau}  \partial_m^2 \log Z \Big|_{m = 0}\,,
 }
where $\tau = \tau_1 + i \tau_2 \equiv \frac{\theta}{2 \pi} + \frac{4 \pi i}{g_\text{YM}^2}$ is the complexified gauge coupling.   As mentioned above, this leads to an explicit calculation of  various terms in the double expansion in $1/N$ and $1/\lambda$ of the CFT correlation function \cite{Binder:2019jwn, Chester:2019pvm}, and these terms reproduced the analogous terms in the  low energy expansion of the type IIB superstring  tree amplitude that are perturbative in $g_s$.

\subsubsection*{$SL(2,\Z)$ duality of CFT correlation functions and type IIB graviton amplitudes}

These connections between $\cN=4$ SYM and type IIB superstring perturbation theory are part of a much richer non-perturbative story that  incorporates the constraints of S-duality. In particular, understanding how the $SL(2,\Z)$ S-duality of the type IIB superstring theory \cite{Hull:1994ys}  arises as the image of   Montonen-Olive $SL(2,\Z)$ duality of  ${\cal N} = 4$ super-Yang-Mills theory \cite{Montonen:1977sn, Witten:1978mh,  Osborn:1979tq}  involves understanding the holographic connection between Yang-Mills instantons and D-instantons  \cite{Banks:1998nr,Bianchi:1998nk,Dorey:1999pd,Dorey:2002ik}, which will be a crucial feature of this paper.
 
Before we go into more details about the $SL(2,\Z)$ invariance properties of the IIB amplitudes and corresponding SYM correlators, we note that neither $\cN=4$ SYM nor the IIB theory on $AdS_5\times S^5$ are entirely $SL(2,\Z)$ invariant. As explained in \cite{Aharony:2013hda}, specifying the SYM theory requires knowing the global form of the gauge group, as well as a discrete theta angle, and general $SL(2, \Z)$ transformations change both the global structure of the gauge group and the discrete theta angle.\footnote{For the $SU(N)$ cases, these theories are labeled as $(SU(N)/\Z_k)_n$ where $k$ is a positive divisor of $N$ and $n$ a $\mZ_k$-valued theta angle \cite{Aharony:2013hda}. 
Under the $S$ generator of $SL(2, \Z)$, the $SU(N)$ theory (no discrete theta angle for this case) is mapped to $(SU(N)/\Z_n)_0$ and vice versa. On the bulk side, the type IIB string theory contains a nontrivial topological sector on AdS$_5$ described by a Chern-Simons-like theory involving the NS and RR 2-form fields  \cite{Witten:1998wy}. The discrete data involved in specifying the boundary SYM theory translates into a choice of boundary conditions for the bulk topological theory, and such boundary conditions transform nontrivially under $SL(2,\Z)$  \cite{Witten:1998wy,Gaiotto:2013nva}.  These topological subtleties are important for understanding the $SL(2,\Z)$ properties of extended objects (such as line operators) or when considering topologically nontrivial backgrounds (such as nontrivial $H^2(M_4,\mZ_N)$ on the 4d boundary manifold $M_4$).}   However,  local correlators in ${\cal N} = 4$ SYM theory  (and correspondingly the amplitudes in IIB string theory) are insensitive to such subtleties, and it is in this sense that we are exploring their $SL(2,\Z)$ invariance properties in this paper.

 \bigskip
 \noindent {\it $SL(2,\Z)$ and the IIB superstring low energy expansion}
 \medskip
 
   The exact coefficients of  higher-derivative interactions in the  low energy expansion of the four-graviton amplitude in IIB superstring theory are $SL(2,\Z)$-invariant functions of the complex string coupling $\tau_s = \chi_s + i / g_s$ that  have been explicitly determined  up to order $D^6R^4$.  These interactions preserve a fraction of the 32 supersymmetries (they are so-called F-terms),  and their form is  severely constrained  by supersymmetry combined with S-duality.  
     For example, the coefficient of the 1/2-BPS $R^4$ interaction is proportional to a non-holomorphic Eisenstein series $E( \frac 32,\tau, \bar \tau)$ \cite{Green:1997tv,Green:1997as,Green:1998by}, which will be defined in Appendix~\ref{app:Einstein}.   When expanded at small string coupling, this Eisenstein series has two terms that are power behaved in $g_s$, corresponding to genus-$0$ and genus-$1$  contributions in string perturbation theory.  In addition, there is an infinite sequence of exponentially  suppressed   non-perturbative terms due to D-instanton effects, where the contribution of a charge-$k$ D-instanton is proportional to $e^{-2\pi k/g_s}$.  Similar comments apply to the coefficient of  the $1/4$-BPS $D^4R^4$ interaction, which has a coefficient proportional to $E(\frac 52, \tau,\bar\tau)$ \cite{Green:1999pu}. Whereas the Eisenstein series satisfy Laplace eigenvalue equations, the coefficient of the $1/8$-BPS interaction, $D^6R^4$, is a novel modular function that satisfies an inhomogeneous Laplace eigenvalue equation  \cite{Green:2005ba} that is also reviewed in Appendix~\ref{app:Einstein}.\footnote{These results were rederived in \cite{Wang:2015jna} using on-shell superamplitude methods \cite{Elvang:2010jv}. Furthermore, more general F-terms involving higher-point  interactions have also been determined \cite{Green:2019rhz}.}

 \bigskip
 \noindent {\it $SL(2,\Z)$ and correlation functions in $\cN=4$ SYM}
 \medskip

In the usual 't Hooft limit, the 't Hooft coupling  $\lambda$ is kept fixed as $N\to \infty$, which requires $g_{\text{YM}}$ to be small.  However, this is incompatible with $SL(2,\Z)$ duality, which has an action on the complex coupling $\tau \equiv \frac{\theta}{2 \pi} + \frac{4 \pi i}{g_\text{YM}^2}$ given by 
\es{sl2act}{\tau \to \tau'= \frac{a\tau+b}{c\tau+d}\,,}
where $a,b,c,d \in  \Z$ and $ad-bc=1$.  In particular, this transformation mixes weak coupling and strong coupling effects.
Therefore, in order to consider the action of $SL(2,\Z)$ on correlation functions in the large-$N$ limit, it is necessary to consider the limit in which $g_{\text{YM}}$ is fixed as $N \to \infty$.  Such a limit had been considered in \cite{Binder:2019jwn}, where it was referred to as the ``very strong coupling limit," and it was also considered in an analogous context in  \cite{Basu:2004dm}.\footnote{Non-'t Hooft limits have been considered before in other gauge theories.  For instance, in the 3d $U(N)_k\times U(N)_{-k}$ ABJM theory \cite{Aharony:2008ug}, the limit in which $k$ is fixed while $N \to \infty$ corresponds to M-theory on $AdS_4 \times S^7/\Z_k$.  In the same theory, a different non-'t Hooft limit, namely that in which $N \to \infty$ with finite $\mu\equiv N/k^5$ considered in \cite{Binder:2019mpb} is somewhat similar to the very strong coupling limit considered here.}   In particular, instanton effects of order  $e^{-8 \pi^2 k /g^2_{\text{YM}}}=e^{-8 \pi^2 k N/\lambda}$, where $k$ (the instanton number) is a positive integer, survive this limit, whereas they  are exponentially suppressed in the usual 't Hooft limit.  In the very strong coupling limit, it is the terms of order $N^{1/2}$, $N^{-1/2}$, and $N^{-1}$ that correspond to the $R^4$, $D^4 R^4$, and $D^6 R^4$ mentioned above.

Using a similar strategy to that outlined above in the 't Hooft limit, we find that the analytic bootstrap constraints combined with the supersymmetric localization constraints coming from \eqref{MixedDer} yield, in the very strong coupling limit, the same Eisenstein series that appear in the low energy expansion of the type IIB four-graviton amplitude.  Indeed, the constraint from \eqref{MixedDer} is sufficient to determine the coefficient of $N^{1/2}$ in the large-$N$ expansion of the CFT four-point function, and, as we will show, this coefficient ends up being proportional to $E(\frac 32, \tau,\bar\tau)$, with precisely the right proportionality factor to match the corresponding term in the superstring amplitude.\footnote{It is worth noting that the fact that the four-point functions of operators in the stress tensor multiplet of ${\cal N} = 4$ SYM, and in particular the quantity \eqref{MixedDer}, are $SL(2, \Z)$ modular invariants is no surprise.  Indeed,  the operators belonging to the stress tensor multiplet transform with well-defined holomorphic and anti-holomorphic  modular weights  $(w,-w)$  under the action of $SL(2,\Z)$ \cite{Intriligator:1998ig,Intriligator:1999ff,Basu:2004dm}. The bottom component of the multiplet we consider has the weight $w=0$,  therefore the corresponding correlators are  $SL(2, \Z)$ invariant. }  This is a full non-perturbative precision test of AdS/CFT!

Note that those contributions that are perturbative in $1/\lambda$ in the  't Hooft limit are also power-behaved in  $g_{\text{YM}}$ in the limit in which $N\to \infty$ and $g_{\text{YM}}$ is fixed.  Thus, the perturbative terms evaluated in earlier work reproduce the  two terms that are power-behaved in $g_{\text{YM}}$ in $E(\frac 32, \tau,\bar\tau)$. Our main challenge is to show that the exponentially suppressed terms in the Eisenstein series can also be reproduced from the ${\cal N} = 4$ SYM theory. Using \eqref{MixedDer}, we find that these exponentially suppressed terms come from considering  the $m=0$ limit of the  instanton contributions to the ${\cal N} = 2^*$ partition function.\footnote{While the strict $m=0$ limit of the ${\cal N} = 2^*$ instanton partition function is trivial, the subleading $m^2$ order is not and contributes to \eqref{MixedDer}.}

At order $N^{-1/2}$ in the large-$N$ expansion, the constraint \eqref{MixedDer} is no longer enough to fully determine the CFT correlation function, but we do find that the integrated correlator \eqref{MixedDer} is proportional to the $E(\frac 52, \tau,\bar\tau)$ modular invariant that appears in the flat space limit  of the IIB amplitude.  Combining the integrated constraint with the flat space string theory answer, we  obtain the complete CFT correlator  up to order $N^{-1/2}$.  While the $E(\frac 32, \tau,\bar\tau)$ modular invariant appearing at order $N^{1/2}$ multiplies a single term of schematic form $R^4$ in the effective field theory in AdS, the $E(\frac 52, \tau,\bar\tau)$ invariant multiplies a linear combination of the $D^4 R^4$ interaction and $R^4 / L^4$, where $L$ is the radius of AdS.

At higher orders in the $1/N$ expansion we do not have sufficient information to determine the four-point correlator in ${\cal N} = 4$ SYM\@.  We will nevertheless argue, more conjecturally, that any given order in the large-$N$ expansion of the integrated correlation function \eqref{MixedDer} is a finite linear sum of non-holomorphic Eisenstein series with rational coefficients.

\subsection{Outline}
\label{outline}

The rest of this paper is organized as follows.  In Section~\ref{SETUP} we give a brief review of the 4-point function of the stress tensor superconformal primary operator in ${\cal N} = 4$ SYM, its relation to the string theory scattering amplitude, and the supersymmetric localization constraint coming from the mixed derivative \eqref{MixedDer}.  Section~\ref{EISENSTEINSPHERE} describes the main technical achievement of this paper, which is the evaluation of the instanton contributions to the integrated correlation function.  These contributions are associated with factors of the Nekrasov partition function \cite{Nekrasov:2002qd,Nekrasov:2003rj}  that enter into the localization result for the mass-deformed $S^4$ partition function and contribute to the $\cN=4$ integrated correlator  described by  \eqref{MixedDer}.  
 We will determine the $k$-instanton contributions to this quantity at orders $N^\half$ and $N^{-\half}$, and show that they match, respectively, the $k$th Fourier modes with respect to  $\theta$ of $E(\frac 32, \tau,\bar\tau)$ and $E(\frac 52, \tau,\bar\tau)$.   The analysis of terms that are higher order in $1/N$ in \eqref{MixedDer} is developed further in Section~\ref{MANYEISENSTEINSERIES}.   We end with a discussion of our results in Section~\ref{DISCUSSION}.  Several technical details are relegated to the Appendices.

\section{Four-point function in ${\cal N} = 4$ SYM}
\label{SETUP}

Let us start with a brief review of the setup of the four-point function of the stress tensor superconformal primary operator in SYM, its relation to the 10d IIB flat space graviton S-matrix and constraints from supersymmetric localization.  For more details, see Ref.~\cite{Binder:2019jwn}.  This operator transforms in the ${\bf 20}'$ of the $SO(6)_R$ R-symmetry, and it can be represented as a traceless symmetric tensor $S_{IJ}(\vec{x})$ with $I, J = 1, \ldots, 6$ as $SO(6)_R$ fundamental indices.  In order to avoid a proliferation of indices, it is customary to contract them with null polarization vectors $Y^I$ satisfying $Y \cdot Y \equiv \sum_{I=1}^6 Y^I Y^I = 0$.  Superconformal symmetry implies that the four-point function of the operator $S(\vec{x}, Y) \equiv S_{IJ}(\vec{x}) Y^I Y^J$ takes the form \cite{Eden:2000bk, Nirschl:2004pa}
 \es{FourPoint}{
  \langle S(\vec{x}_1, Y_1) \cdots S(\vec{x}_4, Y_4) \rangle 
   = \frac{1}{\vec{x}_{12}^4 \vec{x}_{34}^4}
    \left[ 
     \vec{\cS}_\text{free} + {\cal T}(U, V) \vec{\Theta} 
    \right] \cdot \vec{\cB} \,,
 }
where $\vec{x}_{ij}  \equiv \vec{x}_i - \vec{x}_j$, and 
 \es{SThetaB}{
  \vec{\cS}_\text{free} &\equiv \begin{pmatrix}
   1 & U^2 & \frac{U^2}{V^2} & \frac{1}{c} \frac{U^2}V & \frac 1c \frac UV  & \frac 1c U
  \end{pmatrix} \,, \\
   \vec{\Theta} &\equiv \begin{pmatrix}
    V & UV & U & U(U- V - 1) & 1 - U - V & V (V - U - 1) 
   \end{pmatrix} \,, \\
   {\cal B} & \equiv \begin{pmatrix}
    Y_{12}^2 Y_{34}^2 & Y_{13}^2 Y_{24}^2 & Y_{14}^2 Y_{23}^2
     & Y_{13} Y_{14} Y_{23} Y_{24} & Y_{12} Y_{14} Y_{23} Y_{34}
      & Y_{12} Y_{13} Y_{24} Y_{34}
   \end{pmatrix} \,.
 }
Here, $c$ is the conformal anomaly coefficient, which for an $SU(N)$ gauge group equals $c = (N^2 - 1)/4$;  the quantities $U \equiv \frac{ \vec{x}_{12}^2 \vec{x}_{34}^2}{ \vec{x}_{13}^2 \vec{x}_{24}^2}$ and $V\equiv \frac{ \vec{x}_{14}^2 \vec{x}_{23}^2}{ \vec{x}_{13}^2 \vec{x}_{24}^2}$ are the usual conformal invariant cross-ratios;  and $Y_{ij} \equiv Y_i \cdot Y_j$ are $SO(6)_R$ invariants.  Importantly, the only non-trivial information in the correlator \eqref{FourPoint} is encoded in a single function of the conformal cross-ratios, $\cT(U,V)$.

More generally, to describe the holographic correlators, it is simpler to work in Mellin space \cite{Mack:2009gy,Mack:2009mi}.  Let us thus define the  Mellin transform $\cM$ of $\cT$ via
 \es{MellinDef}{
  \cT(U, V)
   = \int_{-i \infty}^{i \infty} \frac{ds\, dt}{(4 \pi i)^2} U^{\frac s2} V^{\frac u2 - 2}
    \Gamma \left[2 - \frac s2 \right]^2 \Gamma \left[2 - \frac t2 \right]^2 \Gamma \left[2 - \frac u2 \right]^2
    \cM(s, t) \,,
 } 
where $u \equiv 4 - s - t$.  Crossing symmetry $\cM(s, t) = \cM(t, s) = \cM(s, u)$, as well as the analytic properties of the Mellin amplitude (for a detailed description see \cite{Binder:2019jwn}), restrict $\cM(s, t)$ to have the following $1/c$ expansion at fixed Yang-Mills coupling:
 \es{MExpansion}{
  \cM(s, t) = \frac{\alpha}{(s - 2) (t - 2) (u - 2)} \frac 1c 
  + \frac{\beta}{c^{7/4}} + \frac { \cM_{\text{1-loop}}(s, t)}{c^2}
   + \frac{\gamma_1 (s^2 + t^2 + u^2) + \gamma_2}{c^\frac{9}{4}} + \cdots \,,
 } 
where the coefficients $\alpha$, $\beta$, $\gamma_i$, etc.~are potentially non-trivial functions of $(\tau, \bar \tau)$.  Here, $\cM_\text{1-loop}$ is the regularized supergravity one-loop amplitude that can be found in \cite{Chester:2019pvm} and will not be discussed here;  in this work, we will instead focus mostly on the $1/c^{7/4}$ and $1/c^{9/4}$ terms, corresponding, respectively, to the $R^4$ and $D^4 R^4$ interaction vertices.  As explained in \cite{Binder:2019jwn}, the constant $\alpha$ can be found as follows.  The free theory contribution in \eqref{FourPoint}, when expanded in conformal blocks, contains twist two operators of all spins.  In the interacting $SU(N)$ gauge theory, however, one expects no operators of twist precisely two, except for those operators belonging to the stress tensor multiplet.  Thus, the conformal block decomposition of the second term in \eqref{FourPoint} must cancel, in part, that of the first term.  A careful analysis shows that this requirement implies \cite{Aprile:2017xsp,Binder:2019jwn}
 \es{Gotalpha}{
  \alpha = 8 \,.
 }
Note that this argument relies crucially on the gauge group being $SU(N)$, and not $U(N)$. A $U(N)$ gauge theory contains a free $U(1)$ sector, and the $S \times S$ OPE contains many operators of twist two beyond those in the stress tensor multiplet.

At each order in $1/c$, one can impose constraints on the coefficients $\beta$, $\gamma_i$, etc.~by either comparing with the (super)graviton four-point scattering amplitude in type IIB string theory in the flat space limit or using the quantity \eqref{MixedDer} (or other similar quantities) derived from supersymmetric localization.  Let us first discuss the constraints from the flat space scattering amplitude, and then those from supersymmetric localization.

\newpage

\subsection{Constraints from the flat space limit}

The IIB four-point scattering amplitude of 10d gravitons and superpartners are restricted by supersymmetry to be proportional to a single function $f(\Mands, \Mandt)$
 \es{ScattAmp}{
  {\cal A}(\Mands, \Mandt) = \cA_\text{SG tree} (\Mands, \Mandt) f(\Mands, \Mandt) \,,
 }
where $\cA_\text{SG tree}$ is the tree-level four-point supergravity amplitude,\footnote{This is given by $\D^{16}(Q)\over {\bf stu}$ in the superamplitude notation where $Q$ denotes the 16-component super-momentum variable.  See, for instance, \cite{Boels:2012ie, Wang:2015jna}. In particular, the component corresponding to the four-graviton scattering  is given by ${R^4 \over \Mands \Mandt \Mandu}$, where $R$ here denotes the linearized Riemann curvature tensor.} $\Mands$ and $\Mandt$ are the Mandelstam invariants.  We will also define $\Mandu \equiv - \Mands - \Mandt$.  
In turn, this function has an expansion at small momentum (more correctly, the expansion is for small values of the dimensionless product between momentum and the string length $\ell_s$) of the form
 \es{fDef}{
   f(\Mands, \Mandt) \equiv
    1 + f_{R^4}(\Mands, \Mandt) \ell_s^6 + f_\text{1-loop}(\Mands, \Mandt) \ell_s^8 + f_{D^4 R^4}(\Mands, \Mandt) \ell_s^{10} + \cdots \,.
 } 
Here, the coefficient function that appears at each order in the expansion may be a non-trivial function of the complexified string coupling $\tau_s = \chi_s + i / g_s$.   In fact, the functions $f_{R^4}$ and $f_{D^4 R^4}$ can be written in terms of non-holomorphic Eisenstein series as \cite{Green:2005ba,Green:1997as,Green:1998by,Green:1999pu}
 \es{fEisenstein}{
  f_{R^4} &= \frac{ \Mands \Mandt \Mandu}{64} g_s^{\frac 32} E(\threeh, \tau_s, \bar \tau_s) \,, \\
  f_{D^4 R^4} &= \frac{\Mands \Mandt \Mandu (\Mands^2 + \Mandt^2 + \Mandu^2)}{2^{11}} g_s^{\frac 52} E(\fiveh, \tau_s, \bar \tau_s) \, .
 }
The Eisenstein series has the following expansion at small $g_s$
(see Appendix~\ref{app:Einstein} for details) 
 \es{EisensteinExpansion}{
  E(r, \tau_s, \bar \tau_s)
   &= \frac{2 \zeta(2 r)}{g_s^r} + 2 \sqrt{\pi} g_s^{r-1} \frac{\Gamma(r - \frac 12)}{\Gamma(r)} \zeta(2r-1) \\
    &{}+ \frac{2 \pi^r}{\Gamma(r) \sqrt{g_s}} \sum_{k\ne 0} \abs{k}^{r-\half}
    \sigma_{1-2r}(|k|) \, 
      K_{r - \frac 12} (2 \pi g_s^{-1} \abs{k}) \, e^{2 \pi i k \chi_s} \, ,
 }
 where the divisor sum $\sigma_p(k)$ is defined as $\sigma_p(k)=\sum_{d>0,{d|k}}  d^p$, and $K_{r - \frac 12}$ is the Bessel function of second kind of index $r-1/2$. 

The relation between the function $f(\Mands, \Mandt)$ in \eqref{ScattAmp} and the Mellin amplitude \eqref{MExpansion} is given by the flat space limit formula \cite{Binder:2019jwn}
 \es{FlatLimit}{
  f(\Mands, \Mandt) = \frac{\Mands \Mandt \Mandu}{2^{11} \pi^2 g_s^2 \ell_s^8}
   \lim_{L / \ell_s \to \infty} L^{14}
    \int_{\kappa - i \infty}^{\kappa + i \infty}
     \frac{d \alpha}{2 \pi i }
      e^\alpha \alpha^{-6} \cM\left( \frac{L^2}{2 \alpha} \Mands, \frac{L^2}{2 \alpha} \Mandt \right)  \,,
 }
where $\kappa > 0$.\footnote{When evaluating this integral, it is useful to note that $\int_{\kappa-i \infty}^{\kappa + i \infty} \frac{d \alpha}{2 \pi i} e^{\alpha} \alpha^{-n} = \frac{1}{\Gamma(n)}$.}  This relation, as well as the AdS/CFT dictionary 
  \es{AdSCFTDict}{
    \tau_s = \tau \,, \qquad
     \frac{L^4}{\ell_s^4} = \lambda=g_\text{YM}^2\sqrt{4c+1}
  }
imply that
  \es{betagammaFlat}{
   \beta(\tau, \bar \tau) = \frac{15}{4 \sqrt{2 \pi^3}}  E(\threeh, \tau, \bar \tau) \,, \qquad
   \gamma_1(\tau, \bar \tau) = \frac{315}{128 \sqrt{2 \pi^5}} E(\fiveh, \tau, \bar \tau) \,, \qquad
   \text{etc.}
  }

\subsection{Constraints from supersymmetric localization}

As explained in \cite{Binder:2019jwn}, supersymmetric localization also imposes constraints on the coefficients of the expansion in \eqref{MExpansion}.  While there are several possible supersymmetric localization constraints, the one studied in \cite{Binder:2019jwn} came from the mixed derivative $\frac{\partial^4 \log Z}{ \partial  \tau  \partial \bar \tau  \partial m^2} \Big|_{m=0}$ of the ${\cal N} = 2^*$ theory on a round $S^4$.

In a large $c$ expansion, a careful analysis of the integrated constraints that follow from $\frac{\partial^4 \log Z}{\partial \tau \partial \bar  \tau   \partial m^2 } \Big|_{m=0}$ as well as the ansatz \eqref{MExpansion} gives \cite{Binder:2019jwn}
 \es{IntConstraint}{
  \frac{   \partial_\tau \partial_{\bar \tau} \partial_m^2 \log Z}
   {\partial_\tau \partial_{\bar \tau}  \log Z} \Bigg|_{m=0}
    = 2 -\frac{ \beta(\tau, \bar \tau) }{5 c^{3/4}} + \frac{C_\text{1-loop}}{c} 
     -  \frac{16\gamma_1(\tau, \bar \tau) + 7 \gamma_2(\tau, \bar \tau)}{35 c^{5/4}} + \cdots \,,
 }
where $C_\text{1-loop}$ is a constant that depends on the precise form of the $\cM_\text{1-loop}$ amplitude that we will not study here.  
We have normalized the integrated four-point correlators by
 \es{dlogZ}{
  \partial_\tau  \partial_{\bar \tau}  \log Z \Big|_{m=0} =  \frac{c}{2 (\Im \tau)^2} \,.
 }

In the following section, using the matrix model for the $\cN = 2^*$ partition function derived by Pestun \cite{Pestun:2007rz}, we will show that, up to an additive ambiguity that is a sum of holomorphic and anti-holomorphic terms in the complexified coupling,
 \es{dlogZdm2}{
   \partial_m^2 \log Z\big\vert_{m=0} = - (4 c + 1) \log \Im \tau  - \frac{\sqrt{2}}{\pi^{3/2}} E(\threeh, \tau, \bar \tau) c^{1/4} 
    + \frac{3}{16 \sqrt{2 \pi^5}} E(\fiveh, \tau, \bar \tau) \frac{1}{c^{1/4}} + \cdots \,.
 }
(See also Eq.~\eqref{dlogZdm2More} in Section~\ref{MANYEISENSTEINSERIES}.)  From comparing this expression to \eqref{IntConstraint} and \eqref{dlogZ}, we then derive constraints on the coefficients appearing in the Mellin amplitude \eqref{MExpansion}.  In particular, at order $c^{-3/4}$, Eqs.~\eqref{IntConstraint}--\eqref{dlogZdm2} imply that $\beta$ should take precisely the same value as derived in \eqref{betagammaFlat} using the flat space limit constraint!  This can be viewed as a derivation of the expression for $f_{R^4}$ given in \eqref{fEisenstein} or as a precision test of AdS/CFT at finite $g_s$.  At order $c^{-5/4}$, we can combine \eqref{IntConstraint}--\eqref{dlogZdm2} with the constraint \eqref{betagammaFlat} obtained from the flat space limit to deduce that 
 \es{gammaRelation}{
  \gamma_2 = -3 \gamma_1 \,.
 }
What remains to be done is to derive Eq.~\eqref{dlogZdm2}, which we carry out in the following section.

\section{Eisenstein series from the mass-deformed $S^4$ partition function}
\label{EISENSTEINSPHERE}

\subsection{Setup}

As shown in \cite{Pestun:2007rz}, up to an overall normalization constant, the mass-deformed partition function of the ${\cal N} = 4$ $SU(N)$ SYM theory (preserving an $\cN=2$ subalgebra) is
 \es{ZFull}{
  Z(m, \tau, \bar \tau) = \int d^{N-1} a \,  \frac{ \prod_{i < j}a_{ij}^2 H^2(a_{ij})}{ H(m)^{N-1} \prod_{i \neq j} H(a_{ij}+ m)}e^{- \frac{8 \pi^2}{g_\text{YM}^2} \sum_i a_i^2} \abs{Z_\text{inst}(m, \tau, a_{ij})}^2 \,,
 }
where $a_{ij}\equiv a_i-a_j$, the integration is over $N$ real variables $a_i$, $i = 1, \ldots, N$, subject to the constraint $\sum_i a_i = 0$, $H(z)$ is the product of two Barnes G-functions, and $Z_\text{inst}$ is the contribution from instantons localized at the poles of $S^4$ that we will come to shortly.  The constrained integral over the $a_i$ can be implemented, for instance, by an integral over $N$ unconstrained $a_i$'s with a $\delta(\sum_i a_i)$ insertion.  Note that the normalization constant that was dropped from \eqref{ZFull} depends on the radius of the sphere, as required by the existence of a conformal anomaly, but it is independent of the coupling $(\tau, \bar \tau)$ and mass $m$.  Consequently, it will drop out of the ratio \eqref{IntConstraint} that is used to related the sphere partition function to the four-point correlator of the ${\bf 20}'$ operator at separated points.

Evaluating $Z(m, \tau, \bar \tau)$ for all $m$ seems to be a complicated task, and we will not pursue it here in full generality.  Instead, what we need is to evaluate the second mass derivative at zero mass, $\partial_m^2 Z(m, \tau, \bar \tau) \big|_{m=0}$.  Let us write the instanton partition function as
 \es{ZInstSum}{
  Z_\text{inst}(m, \tau,  a_{ij}) = \sum_{k=0}^\infty e^{2 \pi i k \tau} Z_\text{inst}^{(k)} (m, a_{ij}) \,,
 } 
with $Z_\text{inst}^{(k)} (m, a_{ij})$ representing the contribution of the $k$-instanton sector and normalized such that  $Z_\text{inst}^{(0)} (m,  a_{ij}) = 1$.  Notably, $Z_\text{inst}(0, \tau, a_{ij}) = 1$ \cite{Pestun:2007rz} so the instantons do not contribute to the sphere partition function at the conformal point.  Then one can argue that
 \es{derMass}{
& \partial_m^2 \log Z\big\vert_{m=0} =\partial_m^2 \log Z\big\vert_{m=0}^\text{pert} +\partial_m^2 \log Z\big\vert_{m=0} ^\text{inst}\,,\\
   &\partial_m^2 \log Z\big\vert_{m=0}^\text{pert}\equiv  \left \langle 
    \partial_m^2 \prod_{i< j} \frac{ H^2(a_{ij})}{H(a_{ij} - m) H(a_{ij} + m)}
   \right\rangle \bigg|_{m=0}\,,\\
   & \partial_m^2 \log Z\big\vert_{m=0} ^\text{inst}\equiv \sum_{k = 1}^\infty  (e^{i k \theta} + e^{- i k \theta})  e^{- \frac{8 \pi^2 k}{g_\text{YM}^2} } \left \langle   \partial_m^2 Z_\text{inst}^{(k)} (m, a_{ij}) \right \rangle \bigg|_{m=0} \,,
 }
where the expectation value is defined to be in the Hermitian matrix model at $m=0$.\footnote{We could consider the expectation values in \eqref{derMass} in either the $SU(N)$ or the $U(N)$ theories, whose partition functions are
\es{Zfree}{
	 Z^{SU(N)} \big \vert_{m=0}
	  =  \int d^{N} a\, \delta\left(\sum_i a_i \right) e^{-\frac{8 \pi^2  }{g_\text{YM}^2} \sum_i a_i^2}  \prod_{i < j} a_{ij}^2
	 \,, \qquad
   Z^{U(N)}\big\vert_{m=0}
	= \int d^{N} a\, e^{-\frac{8 \pi^2  }{g_\text{YM}^2} \sum_i a_i^2}  \prod_{i < j} a_{ij}^2 \,.
	}
Indeed, for any function $F$ that depends only on the differences $a_{ij}$ one can show that 
 $$
  \int d^N a \, \delta\left(\sum_i a_i \right) e^{- \frac{8 \pi^2}{g_\text{YM}^2} \sum_i a_i^2} F(a_{ij})
   = \sqrt{\frac{8\pi}{g_\text{YM}^2 N}} \int d^N a \, e^{- \frac{8 \pi^2}{g_\text{YM}^2} \sum_i a_i^2} F(a_{ij}) \,.
 $$
Thus, even in the presence of an insertion depending only on $a_{ij}$ (as is the case in \eqref{derMass}), the partition functions for the $SU(N)$ and $U(N)$ theories differ by a multiplicative constant that is independent of the operator being inserted.   It follows that normalized expectation values are equal in the $SU(N)$ and $U(N)$ theories. \label{UNFootnote}
} 
The perturbative terms $\partial_m^2 \log Z\big\vert_{m=0}^\text{pert}$ were shown in \cite{Chester:2019pvm} to take the form\footnote{These expressions were given in Eq. (3.1) and (3.20) of \cite{Chester:2019pvm} in the strong coupling expansion, which we can simply convert to the very strong coupling expansion by replacing $\lambda\to g_\text{YM}^2N$.}
\es{pert1}{
 \partial_m^2 \log Z\big\vert_{m=0}^\text{pert}=2N^2 \log g_\text{YM} +\sqrt{N}\left[\frac{16\zeta(3)}{g_\text{YM}^3}+\frac{g_\text{YM}}{3}\right]-\frac{1}{\sqrt{N}}\left[\frac{12\zeta(5)}{g_\text{YM}^5}+\frac{g_\text{YM}^3}{1440}\right]+\dots \,,
}
where further perturbative terms will be discussed in Section~\ref{MANYEISENSTEINSERIES}. These perturbative terms match those of the expected Eisenstein series in \eqref{dlogZdm2} as defined in \eqref{EisensteinExpansion}. To similarly match the instanton terms, we need to show that at large $N$ we have
  \es{Expectation}{
  e^{- \frac{8 \pi^2 k}{g_\text{YM}^2} } \left \langle   \partial_m^2 Z_\text{inst}^{(k)} (m, a_{ij}) \right \rangle \bigg|_{m=0} 
   =&      -   \frac{ 16  \sqrt{N}}{ g_\text{YM}}  k  \,   \sigma_{-2}(k)\, K_1 ( 8 \pi^2 k / g_\text{YM}^2)  \\
         &{}+ \frac{2  }{g_\text{YM} \sqrt{N} }  k^2 \sigma_{-4}(k)\,K_2  ( 8 \pi^2 k / g_\text{YM}^2) 
       + \cdots \, ,
 }
 As a warm-up, let us start with the one-instanton case $k=1$, and then continue with the case of multiple instantons.

\subsection{One-instanton sector}
\label{1-instanton}

The one-instanton contribution $Z_\text{inst}^{(1)}$ is \cite{Nekrasov:2002qd,Nekrasov:2003rj,Pestun:2007rz} (see also \cite{Russo:2013kea})
 \es{onsInstFull}{
  Z_\text{inst}^{(1)}(m, a_{ij}) = - m^2   \sum_{l = 1}^N \prod_{j \neq l} \frac{(a_l - a_j + i )^2 - m^2}{(a_l - a_j) (a_l - a_j + 2 i)} \,.
 }
Taking two derivatives with respect to the mass and evaluating the result at $m=0$ gives
 \es{oneInst}{
  \partial_m^2 Z_\text{inst}^{(1)}  \bigg|_{m=0}
   =  I_1 \,, \qquad
    I_1 \equiv 
     -2 \sum_{l = 1}^N \prod_{j \neq l} \frac{(a_l - a_j + i )^2}{(a_l - a_j) (a_l - a_j + 2 i)} \,.
 }
The quantity $I_1$ can be written in terms of a contour integral as 
 \es{I1Againz}{
  I_1 = 4 \int \frac{dz}{2 \pi} \left[  \prod_j \frac{(z - a_j)^2}{(z - a_j)^2 + 1} - 1 \right]
   = 4  \int \frac{dz}{2 \pi}  \left[ \exp \left( \sum_j \log \frac{(z - a_j)^2}{(z - a_j)^2 + 1} \right) - 1 \right] \,,
 }
where the integration contour is a counter-clockwise contour surrounding the poles at $z = a_j + i$.  Note that the subtraction of $1$ from the integrand does not contribute to the final result, but it does make the integrand decay as $1/z^2$ at $\abs{z} \to \infty$.  Thus, the integration contour in \eqref{I1Againz} can be taken to be the real line.   

We need to evaluate the expectation value of $I_1$ in the Hermitian matrix model \eqref{Zfree}.  In the saddle point approximation, where at leading order in $1/N$ correlation functions factorize, we can approximately write the expectation value of the second expression in \eqref{I1Againz} as
 \es{I1ExpSimple}{
  \langle I_1 \rangle \approx   4 \int_{-\infty}^\infty \frac{dz}{2 \pi}  \left[   \exp \left \langle \sum_j \log \frac{(z - a_j)^2}{(z - a_j)^2 + 1} \right \rangle -1 \right] \, .
 }
We can change variables from $z$ to $x$ where $z = x \sqrt{\lambda}/ (2 \pi) = x \frac{\sqrt{N} g_\text{YM}}{2 \pi}$, and rescale the $a_j$'s similarly:  $a_j = \sqrt{ \frac{\lambda}{2 \pi}} b_j$.   It is a standard result on Hermitian matrix models that, at leading order in large $N$, the $b_j$ become dense and their density is described by the Wigner semicircle law as
 \es{Densityb}{
   \rho(b) = \frac{2}{\pi} \sqrt{1 - b^2} \,, \qquad
    b \in [-1, 1] \,,
 }
normalized so that  $\int_{-1}^1 db\, \rho(b) = 1$.  Making the replacement $\sum_j (\cdots) \to N \int db\, (\cdots) $, we can then write \eqref{I1ExpSimple} approximately as 
 \es{I1ExpSimple2}{
  \langle I_1 \rangle \approx  2 \sqrt{N} \frac{g_\text{YM}}{ \pi}  \int_{-\infty}^\infty \frac{dx}{2 \pi}  \left[ \exp 
   \left( -N \int_{-1}^1 db \, \rho(b)  \log \left( 1 + \frac{4 \pi^2 }{N g_\text{YM}^2 (x - b)^2}\right)   \right) -1 \right] \,.
 }
 The leading contribution to the integral is given by taking $N \to \infty$ in the integrand.  In this limit,  $N \log  \left( 1 + \frac{4 \pi^2}{(x-b)^2 g_\text{YM}^2 N } \right) \to \frac{4 \pi^2}{(x-b)^2 g_\text{YM}^2}$, so 
  \es{I1ExpSimple3}{
   \langle I_1 \rangle \approx   2 \sqrt{N} \frac{g_\text{YM}}{ \pi}  \int_{-\infty}^\infty \frac{dx}{2 \pi}  \left[ \exp 
   \left( -\frac{8 \pi}{g_\text{YM}^2}  \int_{-1}^1 db \,    \frac{\sqrt{1 - b^2}}{  (x - b)^2} \right) -1 \right] \,.
  }
Performing the $b$ integral:
 \es{bIntegral}{
  \int_{-1}^1 db \,    \frac{\sqrt{1 - b^2}}{  (x - b)^2}
   = \begin{cases}
   \pi \left( -1 + \frac{1}{\sqrt{1 - x^{-2}}} \right) \,,  & \text{if $\abs{x} > 1$} \,, \\
   \infty \,, & \text{if $\abs{x} < 1$} \,,
   \end{cases} 
 }
we obtain
  \es{I1Again6}{
   \langle I_1 \rangle
   &\approx  2 \sqrt{N} \frac{g_\text{YM}}{ \pi} e^{\frac{8 \pi^2}{g_\text{YM}^2}} \int_{-\infty}^\infty  \frac{dx}{2 \pi}  
      \left[ e^{- \frac{8\pi^2}{g_\text{YM}^2}    \frac{1}{\sqrt{1 - x^{-2}}}   }  \theta(\abs{x}-1) -  e^{- \frac{8\pi^2}{g_\text{YM}^2}  } \right]  
      \,,
 }
where $\theta(x)$ is the Heaviside theta function.

The integrand is an even function of $x$, so we can just integrate from $x = 0$ to $x = \infty$ and multiply the answer by a factor of $2$.  On this interval, we can further change variables from $x$ to $t$, where 
 \es{tDef}{
  t = \frac{x}{\sqrt{x^2 - 1}} \qquad 
   \Longleftrightarrow \qquad
    x = \frac{t}{\sqrt{ t^2 - 1}} \,,
 }
and obtain
 \es{I1Again7}{
  \langle I_1 \rangle
   &\approx 2 \sqrt{N} \frac{g_\text{YM}}{\pi^2} e^{\frac{8 \pi^2}{g_\text{YM}^2}} 
     \left[-e^{- \frac{8\pi^2}{g_\text{YM}^2} } +
    \int_1^\infty \frac{dt}{( t^2 - 1)^{3/2}} \left( e^{- \frac{8 \pi^2}{g_\text{YM}^2} t} 
      - e^{- \frac{8 \pi^2}{g_\text{YM}^2}} 
     \right)  \right] \,.
 }
We can now use the integral representation of the Bessel $K_1$ function,
 \es{IntFormula}{
  \int_1^\infty dt\, \frac{e^{-a t} - e^{-a}}{(t^2 - 1)^{3/2}} = e^{-a} - a K_1(a)  \,,
 }
to finally write the $\sqrt{N}$ term in the large-$N$ expansion of $\langle I_1 \rangle$ as
 \es{I1Final}{
   \langle I_1 \rangle \big|_{\sqrt{N}} 
   &= -  \sqrt{N}  \frac{16}{g_\text{YM}} e^{\frac{8 \pi^2}{g_\text{YM}^2}} 
      K_1( 8 \pi^2 / g_\text{YM}^2) \,.
 }
Combining with \eqref{oneInst}, this expression implies
 \es{ZRatio}{
  e^{- \frac{8 \pi^2 }{g_\text{YM}^2} } \left \langle  \partial_m^2 Z_\text{inst}^{(1)} (m, a_{ij}) \right \rangle \bigg|_{m=0}  \approx 
   -   \sqrt{N}   \frac{16 K_1 ( 8 \pi^2 / g_\text{YM}^2)}{g_\text{YM}} \,,
 }
in agreement with the expansion of the Eisenstein series---See Eq.~\eqref{Expectation} in the case $k=1$.

Obtaining the term that scales as $1/\sqrt{N}$ in \eqref{Expectation} is not any harder, because this term is suppressed only by a factor of $1/N$ relative to the term we just computed, while the corrections to the approximations made in writing \eqref{I1ExpSimple} and \eqref{I1ExpSimple2} are suppressed by $1/N^2$.  Thus, the next term in \eqref{I1Final} can be obtained by simply expanding \eqref{I1ExpSimple2} to one more order in $1/N$ so that
 \es{NlogRepl}{
  N \log  \left( 1 + \frac{4 \pi^2}{(x-b)^2 g_\text{YM}^2 N } \right) = \frac{4 \pi^2}{(x-b)^2 g_\text{YM}^2}
   - \frac{8 \pi^4}{(x-b)^4 g_\text{YM}^4 N} + \cdots \,,
 }
and evaluating the effect of the second term in \eqref{NlogRepl} in the same way as above.    Plugging \eqref{NlogRepl} into \eqref{I1ExpSimple2}, we obtain
 \es{deltaI1Exp}{
   \langle I_1 \rangle \big|_{1/\sqrt{N}} =     \frac{1}{\sqrt{N}} \frac{32 \pi^2 }{ g_\text{YM}^3 }  \int_{-\infty}^\infty \frac{dx}{2 \pi}  
   \left( \int_{-1}^1 db \,    \frac{\sqrt{1 - b^2}}{  (x-b)^4} \right)    \exp 
   \left( -\frac{8 \pi}{g_\text{YM}^2}  \int_{-1}^1 db \,    \frac{\sqrt{1 - b^2}}{  (x - b)^2} \right)  
   \,.
  }
Performing the $b$ integrals and the $x \to t$ substitution \eqref{tDef}, we find
 \es{delta1Exp2}{
  \langle I_1 \rangle \big|_{1/\sqrt{N}} =
     \frac{1}{\sqrt{N}} \frac{16 \pi^2 }{ g_\text{YM}^3 }  \int_{1}^\infty dt \, 
       t \sqrt{t^2 - 1}
   e^{ -\frac{8 \pi^2}{g_\text{YM}^2} (t-1)}
    =   \frac{1}{\sqrt{N}} e^{\frac{8 \pi^2}{g_\text{YM}^2}} \frac{ 2 K_2 (8 \pi^2 / g_\text{YM}^2) }{ g_\text{YM} } \,.
 }
Adding this expression to \eqref{I1Final} and using the definition of $I_1$ in \eqref{oneInst}, we conclude that
 \es{ZRatioFinal}{
  e^{- \frac{8 \pi^2 }{g_\text{YM}^2} } \left \langle  \partial_m^2 Z_\text{inst}^{(1)} (m, a_{ij}) \right \rangle \bigg|_{m=0}  = 
   -   \sqrt{N}   \frac{16 K_1 ( 8 \pi^2 / g_\text{YM}^2)}{g_\text{YM}}
    + \frac{1}{\sqrt{N}}  \frac{2 K_2 (8 \pi^2 / g_\text{YM}^2) }{ g_\text{YM} }  + \cO(N^{-1})\,,
 }
which is again in agreement with the expectation \eqref{Expectation}. 

Note that in order to go beyond the first two orders in the $1/N$ expansion, one would have to take into account the $1/N^2$ corrections to the saddle point evaluation of expectation values, which we will do later in Section~\ref{MANYEISENSTEINSERIES}.

\subsection{The $k>1$ instanton sector}
\label{MULTIINSTANTON}

We now consider the instanton sector with general $k>1$.  Recall the instanton partition function of $SU(N)$ $\mathcal{N}=2^*$ SYM on $S^4$ can be written as in Eq.~\eqref{ZInstSum}.  The $k$-instanton partition partition $Z^{(k)}_{\rm inst}(m, a_{ij})$ may be further expressed\footnote{Note that the instanton partition function for $\cN=2^*$ SYM was originally obtained for $U(N)$ gauge group \cite{Nekrasov:2002qd,Nekrasov:2003rj}. Later in \cite{Alday:2009aq,Alday:2010vg}, the $SU(N)$ instanton partition function was obtained by factorizing out $U(1)$ contributions  $Z_{U(1)}$ motivated by the AGT correspondence \cite{Alday:2009aq}. Since $Z_{U(1)}$ is holomorphic in $\tau$ and independent of the $a_i$, this is absorbed into the  the holomorphic (anti-holomorphic) ambiguity of \eqref{dlogZdm2}, and does not affect the physical four-point-functions. For this reason, we will simply use the results for $U(N)$ instanton partition functions in the following. We thank Yuji Tachikawa for pointing out the reference \cite{Alday:2010vg}.}  as a sum of contour integrals around poles indicated by a vector $\vec{Y} = \left(Y_1, Y_2, \ldots, Y_N \right)$ of $N$ Young diagrams, 
 \es{ZkSum}{
Z^{(k)}_{\rm inst}(m, a_{ij}) = \sum_{|\vec Y|=k} Z_{k, \vec Y}(m, a_{ij}) \,,
 }
where the Young diagrams $Y_i$ in $\vec{Y}$ are such that the total number of boxes is $k$, namely $|\vec{Y}|\equiv \sum_i |Y_i|=k$.  A Young diagram $Y$ consists of columns of non-increasing height $\lambda_1\geq \lambda_2 \geq \dots$, and the transpose Young diagram $Y^T$ has columns $\lambda^T_1\geq \lambda^T_2\geq \dots$. We will often write them compactly as $Y=[\lambda_1,\lambda_2,\dots]$ and $Y^T=[\lambda_1^T,\lambda_2^T,\dots]$.

Explicitly, $Z_{k, \vec Y}(m, a_{ij})$ is given by \cite{Nekrasov:2002qd,Nekrasov:2003rj}\footnote{Our notation here is related, for instance, to equation (3) of \cite{Okuda:2010ke} by sending $\phi_I \rightarrow -i \phi_I, a_j \rightarrow -i a_j$ and $m_{N} \rightarrow im + \epsilon_+/2$.}
  \es{Ik-gen}{
Z_{k, \vec Y}(m, a_{ij})=& \, {1 \over k!} \left( \epsilon_+ (m^2 + \epsilon^2_-/4) \over 
\epsilon_1 \epsilon_2 (m^2 + \epsilon^2_+/4) \right)^k
\oint \prod_{I=1}^k {d\phi_I\over 2\pi}
\prod_{i=1}^N {(\phi_I -a_j)^2 - m^2  \over (\phi_I  -a_j)^2 +  \epsilon_+^2/4}
\\
&\times
\prod_{I<J}^k { \phi_{IJ}^2[\phi_{IJ}^2 + \epsilon_+^2 ] [\phi_{IJ}^2+(i m-\epsilon_-/2)^2][\phi_{IJ}^2+(i m+\epsilon_-/2)^2]\over [\phi_{IJ}^2 + \epsilon^2_1][\phi_{IJ}^2 + \epsilon^2_2]
[\phi_{IJ}^2+(i m+\epsilon_+/2)^2][\phi_{IJ}^2+(i m-\epsilon_+/2)^2]} \,,
}
where $\phi_{IJ} \equiv \phi_{I}-\phi_{J}$ and $\epsilon_\pm \equiv \epsilon_1\pm \epsilon_2$, with $\epsilon_{1,2}$ being the squashing parameters of $S^4$.  (We have $\epsilon_1 = \epsilon_2 = 1$ for a round sphere.) The multi-dimensional integration contour is determined by the Jeffrey-Kirwan (JK) prescription \cite{1993alg.geom..7001J}, which selects the poles to encircle  based on the choice of an auxiliary vector $\eta$ in $\mR^k$.\footnote{The physics derivation of this contour choice from localization was given in \cite{Hwang:2014uwa}.} As mentioned above, the set of relevant poles are in one-to-one correspondence with the vector of Young diagrams $\vec Y=(Y_1,Y_2,\dots,Y_N)$. Each box in the $j$th Young diagram $Y_j$ is labelled by its position $(\A,\B)$ for positive integers $\A$ and $\B$ denoting the column and the row, respectively, as measured from the bottom left corner (see Figure~\ref{fig:YTcoords}).\footnote{We use the French (as opposed to the English) convention for drawing Young diagrams, where the diagrams extend up and to the right.} The integral then consists of oriented contours surrounding the poles at
\ie \label{YTpoles}
\{\phi_I |1\leq I \leq k \}=\{  a_j + i \epsilon_+/2 + (\A-1)i  \epsilon_1  + (\B-1)i \epsilon_2 ~| (\A ,\B )\in Y_i ,~1\leq i \leq N\} \, .
\fe

For the case of the round $S^4$, which is relevant for our consideration, we set $\epsilon_1=\epsilon_2 =1$.  The $k$-instanton partition function then reduces to
\ie \label{kInstFull}
Z_{k, \vec Y}(m,a_{ij})=&\, {1\over k!} \left( 2m^2\over m^2+1\right)^k
\oint \prod_{I=1}^k {d\phi_I\over 2\pi}
\prod_{i=1}^N {(\phi_I-a_i)^2-m^2\over (\phi_I-a_i)^2+1}
\\
&\times
\prod_{I<J}^k { \phi_{IJ}^2(\phi_{IJ}^2+4)(\phi_{IJ}^2-m^2)^2\over (\phi_{IJ}^2+1)^2 [(\phi_{IJ}-m)^2+1][(\phi_{IJ}+m)^2+1 ]} \,,
\fe 
and the contours now are surrounding 
\ie \label{YTpoles2}
\{\phi_I |1\leq I \leq k \}=\{  a_j  + (\A+\B -1)i  ~| (\A,\B)\in Y_i ,~1\leq i \leq N\} \,. 
\fe
Note some of the simple poles in (\ref{YTpoles}) for certain $\vec{Y}$ are degenerate when $\epsilon_1=\epsilon_2$, giving rise to higher order poles in \eqref{kInstFull}. For certain purposes, (\ref{Ik-gen}) is more convenient to use because it contains only simple poles when $\epsilon_1 \neq \epsilon_2$, and one can set $\epsilon_1=\epsilon_2 =1$ after evaluating the residues. 

In the case of $k=1$, (\ref{kInstFull}) reduces to Eq.~(\ref{onsInstFull}) that we studied in the previous sections.  In this section, we are interested in the cases with $k>1$. For instance, for $k=2$, the relevant poles are enumerated by the Young diagrams $N$-tuples with a total of two boxes:
\ie \label{YT-k2}
\vec Y=\{ \dots,
\yng(2),\dots,
\}, \quad
\vec Y=\{ \dots,
\yng(1,1),\dots,
\}, \quad
\vec Y=\{ \dots,
\yng(1),\dots,\yng(1),\dots
\} \,,
\fe
and the $2$-instanton partition function is then obtained by the sum \eqref{kInstFull} over the contributions from each vector of Young diagrams shown in (\ref{YT-k2}). In general, such a $k$-instanton partition function is very complicated and difficult to analyze, especially when $k$ is large.  However, as we will show below, at the order $m^2$ in a small $m$ expansion (relevant for the four-point correlation function we consider), the $k$-instanton partition function simplifies greatly even for arbitrary $k$.

\subsubsection{Dominance of rectangular Young diagrams}

In this section, we will study the instanton partition function on $S^4$ in the small-mass expansion. It turns out to be convenient to keep $\epsilon_1$ and $\epsilon_2$ general at first and send $\epsilon_1 \rightarrow \epsilon_2$ at the end.\footnote{The same final result can be reached if we instead used (\ref{kInstFull}) where $\epsilon_1 = \epsilon_2$ from the beginning.}   We will argue that the $k$-instanton partition function at order $m^2$, i.e. 
\ie \label{msquare}
I_k \equiv   \partial_m^2 Z^{(k)}_{\rm inst}(m,a_{ij}) \big|_{m=0}\, ,
\fe
only receives  contributions from a vector of Young diagrams   $\vec{Y}$ where all the $N$ diagrams are empty except for one, say $Y_{\hat i}$, which is either rectangular or related by a sequence of {\it partial transpositions} to a rectangular diagram (we define what we mean by partial transpositions below).  In other words, $\vec{Y}$ is related by partial transpositions to the following
\ie \label{rectangular}
Y_i=&\varnothing \qquad {\rm if} \qquad i\neq \hat i \, ,
\\
Y_{\hat i}=&Y_{p\times q} \quad {\rm for} \quad p\leq q\in \mZ_+,~~pq=k \,,
\fe
where $Y_{p\times q}$ denotes the rectangular Young diagram with $p$ rows and $q$ columns. See Appendix~\ref{app:rectd} and Theorem~\ref{thmvan} for the complete proof. In the following we present a brief summary of the proof.

Given a Young diagram $Y=[\lambda_1,\dots,\lambda_l]$, the {\it partial transposition} at horizontal position $\A$, denoted by PT$_\A$, is a local operation on $Y$ that replaces the rightmost block (subdiagram) $P=[\lambda_\A,\lambda_{\A+1},\dots,\lambda_l]$ by its transpose, while leaving the rest of the  Young diagram unchanged, provided that the new diagram is still a Young diagram.  Otherwise, the particular partial transposition PT$_\A$ is defined to act trivially on $Y$ (for a more detailed discussion, see Appendix~\ref{app:rectd}).  In particular PT$_1$ generates the usual transposition of Young diagrams.
In general $\vec Y$ can have many different partners from {\it partial transpositions}.  For instance, for $k=4$, we have two cases of vectors $\vec{Y}$ of the type \eqref{rectangular}: 
\ie \label{k4YT}
\vec Y=\{ \dots,
\yng(4),\dots
\},\qquad
\vec Y =\{ \dots,
\yng(2,2),\dots
\}\,.
\fe
The other Young diagram vectors $\vec Y$ that contribute are either of the form
\ie \label{k4YT1}
\vec Y=\{ \dots,
\yng(1,1,1,1),\dots
\}\, ,
\fe
and related by a single transposition to the first case in \eqref{k4YT}, or are of the form
\ie \label{k4YT2}
\ytableausetup
{boxsize=1.25em}
\ytableausetup
{aligntableaux=bottom}
\vec Y=\{ \dots,
\ydiagram{1,1,2}*[*(orange)]{0+1,0+1}
,\dots
\},\qquad
\vec Y=\{ \dots,
\ydiagram{1,1}*[*(orange)]{1+0,1+2}
,\dots
\}\, .
\fe
The two $\vec{Y}$'s in \eqref{k4YT2} are related to each other by a transposition, and the second one is related to the second diagram in \eqref{k4YT} via a single partial transposition involving the orange block. 

Note that if we set $\epsilon_1 = \epsilon_2 =1$ from the beginning (i.e.~starting with (\ref{kInstFull})), the $\vec{Y}$'s and their partners related by transpositions and partial transpositions degenerate in the sense that they give exactly the same poles in the integrand of the contour integral (see Lemma~\ref{lemmainvset}).

It is useful to introduce the arm-length and leg-length of a box at position $(\A,\B)$  (see Figure~\ref{fig:YTcoords}):  
\ie
h(\A,\B) \equiv \lambda_\B^T - \A\,, \qquad v(\A,\B) \equiv \lambda_\A -\B \,,
\fe
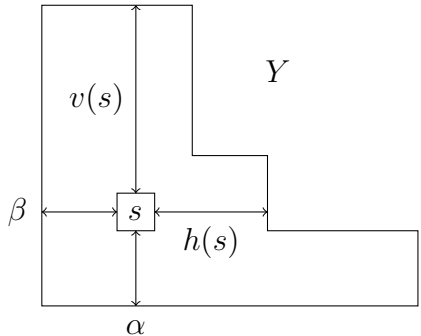
\begin{figure}[htb!]
	\centering
	\begin{tikzpicture}
	\draw (0,0) -- ++(5,0) --++(0,1) -- ++(-2,0) --++(0,1) --++(-1,0)--++(0,2)--++(-2,0) -- cycle
	node[above right=4cm]{$Y$};
	\draw (1,1) -- ++(.5,0) --++(0,.5) -- ++(-.5,0)  -- cycle
	node[above right]{$s$};
	\draw [<->](1.25,1.5) -- ++(0,2.5) node[midway, left=.01cm] {$v(s)$} ;
	\draw [<->](1.25,1.0) -- ++(0,-1) node[near end, below=.3cm] {$\A $} ;
	\draw [<->](1.5,1.25) -- ++(1.5,0) node[midway, below=.05cm] {$h(s)$} ;
	\draw [<->](1.0,1.25) -- ++(-1,0) 
	node[near end, left=.3cm] {$\B $};
	\end{tikzpicture}
	\caption{An example of Young diagram $Y$, the coordinates $(\A,\B)$ for box $s$ and the corresponding arm-length $h(s)$ and leg-length $v(s)$.}
	\label{fig:YTcoords}
\end{figure}
where $\lambda_\beta^T$ is as defined below \eqref{ZkSum}.   In other words, $h(\alpha,  \beta)$ and $v(\alpha, \beta)$ measure the number of boxes from the given box $(\alpha, \beta)$ to the right edge and the top edge, respectively, of the Young diagram.

As we show in Appendix~\ref{app:rectd}, in general the contribution of a given $\vec{Y}$ to the second mass-derivative of the instanton partition function behaves as
\ie \label{ni-counting}
I_{k, \vec{Y}}  \equiv  \partial_m^2 Z_{k, \vec{Y}}(m,a_{ij}) \big|_{m=0} \sim  (\epsilon_-)^{-2+\sum_{i=1}^N \mu(Y_i) } \, ,
\fe
in the limit $\epsilon_- \rightarrow 0$, where
\ie
\mu(Y)=2n_0(Y)-n_{-1}(Y)-n_1(Y)
\fe
with $n_d$ representing the number of the boxes with $h-v =d$.

It is straightforward to see that all rectangular Young diagrams $ {Y}_{p\times q}$ have 
\ie \label{ni-rec}
\mu(Y_{p\times q})=\begin{cases}
 1 \, , \quad {\rm if} \quad  p\neq q \, , 
 \\
 2 \, , \quad {\rm if} \quad  p= q \, ,
\end{cases}
\fe
and, importantly, the operation of partial transposition does not change the value  of  $\mu(Y)$ (see Lemma~\ref{lemmainv}). Therefore, from (\ref{ni-counting}), we see that each $\vec{Y}$ of the type \eqref{rectangular} leads to a finite contribution to the instanton partition function in the small-mass expansion.\footnote{One may worry about the divergence arising from $I_{k, (\dots,Y_{p\times q},\dots)} \sim (\epsilon_-)^{-1}$ for $p\neq q$, but the divergence is canceled out in the sum  $I_{k, (\dots,Y_{p\times q},\dots)}+I_{k, (\dots,Y_{q\times p},\dots)}$ and leads to a finite result.  }

Furthermore, all other types of Young diagram $N$-tuples either have $\sum_{i=1}^N\m(Y_i)>2$, and therefore their contributions manifestly vanish, or $\sum_{i=1}^N\m(Y_i)=2$, but for such a $\vec{Y}$ vector (where every non-empty $Y_i$ is not related to a rectangular diagram by partial transpositions) there is always a cancellation between $\vec{Y}$ and a partner related by a certain involution  that arises from a sequence of partial positions (see Appendix~\ref{app:rectd}).\footnote{Such a cancellation does not happen for \eqref{rectangular} with a square Young diagram  $Y_{p\times p}$, which also have $\mu=2$ as shown in (\ref{ni-rec}).} For instance, consider the following two $\vec{Y}$'s that are related by partial transposition to each another (transposing the orange block):
\ie 
\ytableausetup
{boxsize=1.25em}
\ytableausetup
{aligntableaux=bottom}
\vec Y=\{ \dots,
\ydiagram{1,1,1,1,1}*[*(orange)]{0,0,0,1+1,1+1},\dots
\}\,, \qquad
\vec Y=\{ \dots,
\ydiagram{1,1,1,1,1}*[*(orange)]{0,0,0,0,1+2},\dots
\}\,.
\fe
The above diagrams in $\vec{Y}$'s are not related to any rectangular Young diagrams by partial transpositions, and each has $\mu=2$ (therefore naively would lead to a finite contribution to the instanton partition function). But the finite contributions from these two $\vec{Y}$'s in fact cancel out.  This is a very general phenomenon. Again, we refer a general proof of all these statements to the Appendix~\ref{app:rectd}.

\subsubsection{Instanton partition function at order $m^2$}

After showing the dominance of the rectangular Young diagrams for the leading $m^2$ order of the instanton partition function, let us now compute the residues necessary for evaluating \eqref{Ik-gen} or \eqref{kInstFull}.  
As mentioned above, while in general the contribution of any given Young diagram to the $k$-instanton partition function is rather complicated, the quadratic term in the small mass expansion will end up being quite simple.   

For instance, the instanton partition function for $k=2$ can be computed either from (\ref{Ik-gen}), in which case we have non-trivial contributions only from
\ie \label{k=2-YD}
\vec Y=\{ \dots,
\yng(1,1),\dots
\}\, , \qquad
\vec Y=\{ \dots,
\yng(2),\dots
\}\,.
\fe
Alternatively, we can also use (\ref{kInstFull}), in which case the above two $\vec{Y}$'s give the same set of poles.     For the case of $k=2$, for each Young diagram, there are two ways of distributing $\phi_I$ to the boxes of the Young diagram, which lead to two contributions to the partition function. For instance, using the formula (\ref{Ik-gen}), the two residues that we should evaluate for the first Young diagram in (\ref{k=2-YD}) are given by
\ie
R_1 = {\rm Res}_{\phi_1=a_j+ i \epsilon_+/2} \, {\rm Res}_{\phi_2=\phi_1+ i\epsilon_1} \, , \qquad 
R_2 = {\rm Res}_{\phi_2=a_j+i \epsilon_+/2 } \, {\rm Res}_{\phi_1=\phi_2+ i\epsilon_1} \, .
\fe
For the second Young diagram (which is the conjugate of the first one), the residues are computed in the same way, but with $\epsilon_1 \leftrightarrow \epsilon_2$.  Furthermore, $R_1$ and $R_2$ give the same contributions since they simply exchange $\phi_1 \leftrightarrow \phi_2$.  It is thus convenient to only evaluate the residue ${\rm Res}_{\phi_2=\phi_1+i \epsilon_1}$ in $R_1$, and ${\rm Res}_{\phi_1=\phi_2+ i \epsilon_1}$ in $R_2$, and leave the remaining variable (namely $\phi_1$ in the first case, and $\phi_2$ in the second case) unintegrated.   (The variable that remains is the one corresponding to the box of the Young diagram that sits in the bottom left corner.) Computing the residues explicitly and summing up the contributions from both Young diagrams, we obtain
\ie \label{eq:12-tab}
I_{{1 \times 2}} &=  \oint  {d z \over 2\pi}
\prod_{k_a}
\prod_{j=1}^N {(z -a_j + k_a i )^2\over (z -a_j + k_a i)^2+1 } \\
& ~~~ \times
\left[ 5  + \sum_{j=1}^N
{3 i \over (z-a_j+2 i) (z-a_j+ i) (z-a_j)}
\right] \,, 
\fe
where $k_a =\{0, 1\}$, and we have set $\epsilon_1=\epsilon_2=1$ at the end of the computation and denoted the unintegrated variable as $z$ in (\ref{eq:12-tab}). The contour for the remaining  $z$ integral is around poles at $z=a_j +i$, with $j=1,2, \ldots, N$. Note the same result (\ref{eq:12-tab}) can be obtained by using (\ref{kInstFull}) (i.e. setting $\epsilon_1=\epsilon_2=1$ from the beginning). 

In general, for the case of $k$ instantons, there are $k!$ ways of assigning $\phi_I$'s to a given Young diagram,\footnote{If we use (\ref{kInstFull}) for the computation, one should also divide the symmetry factors due to the degeneracy of the poles (i.e.~according to (\ref{YTpoles2}), the $\phi_I$'s corresponding to the boxes at positions $(\alpha, \beta)$ of a Young diagram with the same values of $\alpha+\beta$  surround  the same poles.).} and we will integrate out all the $(k-1)$ $\phi_I$'s, but again leave the one that is assigned to the bottom left corner box unintegrated (just as the $k=2$ case), and denote it by $z$. The contour for the remaining $z$-integration is then a counter-clockwise contour surrounding the poles at $z=a_j +i$, with $j=1,2, \ldots, N$. 

Let us consider another example before presenting a general $k$-instanton formula. For instance, the $4$-instanton partition function, for which there are two types of  Young diagrams that contribute. The first type is given by
\ie \label{14YT}
\vec Y=\{ \dots,
\yng(1,1,1,1),\dots
\}\, , \qquad
\vec Y=\{ \dots,
\yng(4),\dots
\}\, ,
\fe
and the second kind is
\ie \label{22YT}
\vec Y =\{ \dots,
\yng(2,2),\dots\}\, ,\quad \vec Y=\{ \dots,
\yng(1,1,2),\dots
\},\quad
\vec Y=\{ \dots,
\yng(1,3),\dots
\}\,.
\fe
Computing each contribution explicitly, we again find very compact results with similar structures as those of (\ref{eq:12-tab}) from the two-instanton case. For the Young diagrams in (\ref{14YT}), we find
\begin{align} \label{eq:14-tab}
I_{{1 \times 4}} &=  \oint  {d z \over 2\pi}
\prod_{k_a}
\prod_{j=1}^N {(z -a_j + k_a i )^2\over (z -a_j + k_a i)^2+1 } \\
& ~~~ \times
\left[ {17\over 4}  + \sum_{j=1}^N
{45 i \over 2(z-a_j+4 i) (z-a_j+ 3i) (z-a_j)}
\right] \,, \nonumber
\end{align}
with $k_a =\{0, 1, 2,3\}$, while for (\ref{22YT}) we obtain
\begin{align} \label{eq:22-tab}
I_{{2 \times 2}} &=  \oint  {dz \over 2\pi}
\prod_{k_a}
\prod_{j=1}^N {(z -a_j + k_a i )^2\over (z-a_j + k_a i)^2+1 }  \,,
\end{align}
with $k_a =\{0, 1, 1, 2\}$. 

The above simple structures present in the examples we have studied generalize. Indeed, we  find that the contribution to the $k$-instanton partition function coming from Young diagram vectors of the form \eqref{rectangular}  as well as its partial transpositions is given by 
 \es{eq:pq-tab}{
I_{p \times q} &=  \oint  {dz \over 2\pi}
\prod_{k_a}
\prod_{j=1}^N {(z-a_j + k_a i )^2\over (z-a_j + k_a i)^2+1 } \times
\left[
{4 \over 1 + \delta_{pq}} \left({1\over p^2}+{1\over q^2} \right) \right.
 \\
& \left. + \sum_{j=1}^N
{i f(p, q) \over (z-a_j+(p+q -1) i) (z-a_j+ (q-1) i) (z-a_j+(p-1)i)}
\right] \,,
}
where the integration contour of the left-over $z$ is a counter-clockwise contour surrounding the poles at $z= a_j + i$ (with $j=1,2, \ldots, N$). The $k_a$'s ($k$ of them) are read off from the vector of Young diagrams $\vec Y$ as in \eqref{rectangular}, and they are given by
\begin{align} \label{kas}
k_a =\{0, 1, \cdots, p-1; 1, 2, \cdots, p; \cdots; q-1, q, \cdots, p+q-2 \} \, .
\end{align}
Finally, the function $f(p, q)$ is 
\ie
f(p, q) = {2 (q+p)(q-p)^2 \over pq} \, .
\fe 
This function is symmetric in $p\leftrightarrow q$ and vanishes at $p=q$.  
The formula (\ref{eq:pq-tab}), which is one of our main results, was obtained by studying the pattern of many non-trivial examples.  We will study its large-$N$ expansion in the next section.  For the special case  where the non-trivial rectangular Young diagram in $\vec Y$ is $Y_{1 \times k}$, we provide a proof in Appendix \ref{sec:recursion} using a  recursion relation satisfied by the instanton partition function \cite{Kanno:2013aha, Nakamura:2014nha}.  However, the same recursion relation proof for general $ Y_{p \times q}$ becomes a bit cumbersome. Nevertheless, we have verified the formula (\ref{eq:pq-tab})  explicitly up to $k=20$ instantons. Furthermore,  the result \eqref{eq:pq-tab} clearly has a  structure that is consistent with 
 (\ref{kInstFull}): the constant part ${4 \over 1 + \delta_{pq}} \left({1\over p^2}+{1\over q^2} \right)$ arises only from the second line of (\ref{kInstFull}), whereas the term depending on $z$ and $a_j$ involves expanding the first line of (\ref{kInstFull}) while taking the residues around higher order poles. 

Finally, we remark that given the structure of the $k$-instanton contribution to the non-holomorphic Eisenstein series, especially the appearance of the divisor sum (see (\ref{EisensteinExpansion})), it is not surprising that the relevant $\vec{Y}$'s are only the rectangular ones. As we will show, each  Young diagram  $Y_{p\times q}$ contributes a term  in the divisor sum for a non-holomorphic Eisenstein series (proportional to ${p^{1-2r}}+{q^{1-2r}}$ for $E(r, \tau, \bar \tau)$).\footnote{In particular, the term ${4 \over 1 + \delta_{pq}} \left({1\over p^2}+{1\over q^2} \right)$, which dominates in the large $N$ limit, gives directly the divisor sum $\sigma_{-2}(|k|)$ in $E(\threeh, \tau, \bar \tau)$.} 

\subsubsection{Large-$N$ expansion}

We will now compute the expectation value of $I_{ p \times q}$ in the Hermitian matrix model (\ref{Zfree}), in the large-$N$ expansion. The computation is similar to that in the one-instanton case presented in Section~\ref{1-instanton}, so we will therefore be brief here. In the large-$N$ limit, we have 
\ie \label{eq:pq-tab-2}
\langle I_{p \times q} \rangle & \approx {\sqrt{N} g_{\rm YM} \over 2\pi} \int_{-\infty}^{\infty}  {dx \over 2\pi}
\left( \exp\left[ -N \int^1_{-1} db \rho(b) \sum_{k_a} \log \left( 1+ {4\pi^2 \over N g^2_{\rm YM}  (x - b +{2\pi i \over \sqrt{N} g_{\rm YM} } k_a)^2 }  \right) \right] \right. \cr
& \times 
\left. \left(
{4 \over 1 + \delta_{pq}} {p^2 + q^2\over p^2 q^2}+ i  {({2\pi \over g_{\rm YM}})^3 \over \sqrt{N}} \int_{-1}^1 db \rho(b)
{ f(p, q)  \over g(x, b) } \right) - {4 \over 1 + \delta_{pq}} {p^2 + q^2\over p^2 q^2} \right) \,,
\fe
with $k_a$ given in (\ref{kas}). In the above, we have approximated the sums as integrals, and we have deformed the contour by subtracting an appropriate constant from the integrand.  Again, the density measure $\rho(b)$ obeys Wigner's semi-circle law, given by $\rho(b) = {2\over \pi} \sqrt{1-b^2}$, and 
\ie
g(x, b) = \left[x- b+ {2\pi i \over \sqrt{N} g_{\rm YM} } (p +q- 1) \right] \left[x- b + {2\pi i \over \sqrt{N} g_{\rm YM} } (q-1) \right] \left[ x - b + {2\pi i \over \sqrt{N} g_{\rm YM} } (p-1) \right] \,. 
\fe

Just as the one-instanton case, we expand the integrand in $1/N$, and the integration over $x$ can be separated into different regions: $x \in \{-1,1\}$, $x \in \{-\infty,-1\}$, and $x \in \{1, +\infty\}$.  We find that the leading term in the $1/N$ expansion is given by
\ie
\langle I_{p \times q}\rangle \big{|}_{\sqrt{N}} & ={\sqrt{N} g_{\rm YM} \over 2\pi^2} 
\left(  - {4 \over 1 + \delta_{pq}} {p^2 + q^2\over p^2 q^2} \right) + {\sqrt{N} g_{\rm YM} \over \pi}  \exp\left[ k {8 \pi^2  \over g_{\rm YM}^2 } \right]  
\left( {4 \over 1 + \delta_{pq}} {p^2 + q^2\over p^2 q^2} \right)
 \cr
& \quad \times \int_1^{\infty} {dx \over 2\pi}
\left( \exp\left[ -  {8 k \pi^2  x \over g_{\rm YM}^2 (x^2-1)^{1\over2} }  \right] -   \exp\left[ -k {8 \pi^2  \over g_{\rm YM}^2 } \right] \right) \,.
\fe
After a change of integration variable identical to (\ref{tDef}), it is straightforward to show that $\langle I_{p \times q}\rangle \big{|}_{\sqrt{N}}$ can be expressed in terms of a Bessel $K_1$ function:  
\begin{align}
\langle I_{p \times q}\rangle \big{|}_{\sqrt{N}} = -   {k \over 1 + \delta_{pq}} \left( \frac{1}{p^2} + \frac{1}{q^2} \right)   {16 \sqrt{N}  \over g_{\rm YM} }  \exp\left[ k {8 \pi^2  \over g_{\rm YM}^2 } \right]  
K_1(k {8 \pi^2  \over g_{\rm YM}^2 })  \,.
\end{align}
Summing over all possible rectangular $\vec{Y}$'s for a given $k$-instanton sector (namely all the divisors of $k$), we have
\ie \label{leading-term}
\sum_{pq=k, \, 0< p \leq q}  \langle I_{p \times q}\rangle \big{|}_{\sqrt{N}} = - {16 \sqrt{N}  \over g_{\rm YM} }  k\, \sigma_{-2}(k)  \exp\left[ k {8 \pi^2  \over g_{\rm YM}^2 } \right]  
K_1(k {8 \pi^2  \over g_{\rm YM}^2 }) \,, 
\fe
where we have used
\ie
\sum_{pq=k, \, 0< p\leq q} {1 \over 1 + \delta_{pq}} \left( \frac{1}{p^2} + \frac{1}{q^2} \right) = \sigma_{-2}(k) \,.
\fe

The $N^0$ order term vanishes due to the fact that integrand is odd in $x$. Then, at the next order we have a $1/\sqrt{N}$ term that takes the following form
\begin{align} \label{eq:pq-tab-4}
\langle I_{p \times q} \rangle \big{|}_{1/\sqrt{N}} &= {1\over 1+\delta_{pq}} {16\pi^2 \over \sqrt{N} g^5_{\rm YM}  k^2 }  \exp\left[ k {8 \pi^2  \over g_{\rm YM}^2 } \right]  \int_1^{\infty} {dx }
\left( \exp\left[ -  {8 k \pi^2 x  \over g_{\rm YM}^2 (x^2-1)^{1\over2} }  \right] \right. \cr
& \qquad \times 
\left.  \left[ {c_1 g_{\rm YM}^2 x  \over  (x^2 -1)^{5\over 2} }- {c_2 \pi^2 \over (x^2 -1)^{3}  }   \right]   \right) \,,
\end{align}
where $c_1, c_2$ are given by
\begin{align}
c_1 &= k (p - q)^2 (p + q) (2 p + 2 q-3)  + (k+6 \sum_a k_a^2 ) (p^2 + q^2) \, , \cr
c_2 &= 8 k  (p - q)^2 (p + q) \sum k_a + 16  (p^2 + q^2) (\sum_a k_a)^2 \, ,
\end{align}
with $k_a$ given in (\ref{kas}). 
Again, by a change of integration variable, the integral of $\langle I_{p \times q} \rangle \big{|}_{1/\sqrt{N}}$ reduces to a standard Bessel $K_2$ function, 
\begin{align} \label{eq:pq-tab-5}
\langle I_{p \times q} \rangle \big{|}_{1/\sqrt{N}} &={1\over 1+\delta_{pq}}  {16\pi^2 \over \sqrt{N} g^5_{\rm YM}  k^2 }  \exp\left[ k {8 \pi^2  \over g_{\rm YM}^2 } \right]  
  \left[ {c_1 g_{\rm YM}^2 \over   (k {8 \pi^2  \over g_{\rm YM}^2 })  }-  {3c_2 \pi^2 \over (k {8 \pi^2  \over g_{\rm YM}^2 })^2  }   \right]  K_2(k {8 \pi^2  \over g_{\rm YM}^2 }) \cr
  & ={k^2\over 1+\delta_{pq}}  \left( \frac{1}{p^4} + \frac{1}{q^4} \right)   \,  {2\over \sqrt{N} g_{\rm YM}} \exp\left[ k {8 \pi^2  \over g_{\rm YM}^2 } \right]   K_2(k {8 \pi^2  \over g_{\rm YM}^2 })\, .
\end{align}
Again, taking into account all the relevant contributions from rectangular $\vec{Y}$'s with $k$ boxes, the prefactor in the above formula becomes the divisor sum $\sigma_{-4}(k)$, namely, 
\ie
\sum_{pq=k, \, 0< p\leq q} {1 \over 1 + \delta_{pq}} \left( \frac{1}{p^4} + \frac{1}{q^4} \right) = \sigma_{-4}(k) \,.
\fe
Combining \eqref{eq:pq-tab-5} with the result of the leading large-$N$ term in \eqref{leading-term},  we obtain 
\ie
- {16 \sqrt{N}  \over g_{\rm YM} }  k\, \sigma_{-2}(k)  \exp\left[ k {8 \pi^2  \over g_{\rm YM}^2 } \right]  
K_1(k {8 \pi^2  \over g_{\rm YM}^2 })  +{2\over \sqrt{N} g_{\rm YM}}   k^2 \, \sigma_{-4}(k) \exp\left[ k {8 \pi^2  \over g_{\rm YM}^2 } \right]   K_2(k {8 \pi^2  \over g_{\rm YM}^2 }) \,.
\fe
Therefore we have proven (\ref{Expectation}). In the next section, we will study the higher order terms in the $1/N$ expansion, and show that in fact they are also given by non-holomorphic Eisenstein series.

\section{Eisenstein series at higher orders in $1/N$}
\label{MANYEISENSTEINSERIES}

In this section we will provide additional evidence that the coefficients in the large-$N$ expansion of $\partial_m^2 \log Z \big\vert_{m=0}$, which was derived in the previous sections  to the first couple of orders in $1/N$ in terms of the Eisenstein series shown in \eqref{dlogZdm2}, takes the form of Eisenstein series to all orders in $1/N$. In particular, we propose that, through order $N^{-7/2}$, $\partial_m^2 \log Z \big\vert_{m=0}$ is 
 \es{dlogZdm2More}{
   \partial_m^2 \log Z\big\vert_{m=0} =&2 N^2\log g_\text{YM}-\frac{\sqrt{N}}{\pi^{\frac32}} E( {\scriptstyle {3 \over 2}}, \tau,\bar\tau)+\frac{3}{16\sqrt{N}\pi^{\frac52}}E( {\scriptstyle {5 \over 2}},\tau,\bar\tau)\\
 &+\frac{1}{{N}^{\frac32}}\left[-\frac{13}{2^9 \pi^{\frac32}} E( {\scriptstyle {3 \over 2}},\tau,\bar\tau)+\frac{135}{2^{11}\pi^{\frac72}} E( {\scriptstyle {7 \over 2}},\tau,\bar\tau)\right]\\
 &+\frac{1}{{N}^{\frac52}}\left[-\frac{75}{2^{12} \pi^{\frac52}}E( {\scriptstyle {5 \over 2}},\tau,\bar\tau)+\frac{1575}{2^{14} \pi^{\frac92}} E( {\scriptstyle {9 \over 2}},\tau,\bar\tau)\right]\\
 &+\frac{1}{{N}^{\frac72}}\left[\frac{1533}{2^{18} \pi^{\frac32}}E( {\scriptstyle {3 \over 2}},\tau,\bar\tau)-\frac{80325}{2^{21} \pi^{\frac72}}E( {\scriptstyle {7 \over 2}},\tau,\bar\tau)+\frac{2480625}{2^{23} \pi^{\frac{11}{2}}}E( {\scriptstyle {11 \over 2}},\tau,\bar\tau)\right]\\
 &+O(N^{-\frac92})+\text{(anti)holomorphic ambiguity}\,.
 }
 We can then take derivatives in $\tau$ and $\bar{\tau}$ to obtain the $SL(2,\mathbb{Z})$ invariant quantity   
    \es{dlogZdm2More2}{
 \tau_2^2 \partial_\tau\partial_{\bar\tau} \partial_m^2 \log Z\big\vert_{m=0} =&\frac{ N^2}{4} -\frac{3\sqrt{N}}{2^4 \,\pi^{\frac32}} E( {\scriptstyle {3 \over 2}},\tau,\bar\tau)+\frac{45}{2^8 \sqrt{N}\pi^{\frac52}}E( {\scriptstyle {5 \over 2}},\tau,\bar\tau)\\
 &+\frac{1}{{N}^{\frac32}}\left[-\frac{39}{2^{13} \pi^{\frac32}}E( {\scriptstyle {3 \over 2}},\tau,\bar\tau)+\frac{4725}{2^{15} \pi^{\frac72}}E( {\scriptstyle {7 \over 2}},\tau,\bar\tau)\right]\\
 &+\frac{1}{{N}^{\frac52}}\left[-\frac{1125}{2^{16} \pi^{\frac52}}E( {\scriptstyle {5 \over 2}},\tau,\bar\tau)+\frac{99225}{2^{18} \pi^{\frac92}}E( {\scriptstyle {9 \over 2}},\tau,\bar\tau)\right]\\
 &+\frac{1}{{N}^{\frac72}}\left[\frac{4599}{2^{22} \pi^{\frac32}}E( {\scriptstyle {3 \over 2}},\tau,\bar\tau)-\frac{2811375}{2^{25} \pi^{\frac72}}E( {\scriptstyle {7 \over 2}},\tau,\bar\tau)+\frac{245581875}{2^{27} \pi^{\frac{11}{2}}}E( {\scriptstyle {11 \over 2}},\tau,\bar\tau)\right]\\
 &{}+O(N^{-\frac92})\,.
 }

The first piece of evidence for \eqref{dlogZdm2More} comes from considering the terms that are perturbative in $1/\lambda=1/(g_\text{YM}^2N)$, which as discussed above were computed in \cite{Chester:2019pvm} and take the form\footnote{The expression here includes a further order in $1/N$ and several more orders in $1/\lambda=1/(g_\text{YM}^2N)$ relative to Eqs.~(3.1) and~(3.20) of \cite{Chester:2019pvm}, which can be easily computed using the same methods explained in that work.}
\es{c2}{
 \partial_m^2\log Z\big\vert_{m=0}^\text{pert}=&2 N^2\log g_\text{YM}+\sqrt{N}\left[-\frac{16\zeta(3)}{g_\text{YM}^3}-\frac{g_\text{YM}}{3}\right]+\frac{1}{\sqrt{N}}\left[\frac{12\zeta(5)}{g_\text{YM}^5}+\frac{g_\text{YM}^3}{1440}\right]\\
 &+\frac{1}{N^{\frac32}}\left[\frac{135 \zeta (7)}{8 g_\text{YM}^7}+\frac{g_\text{YM}^5}{215040}-\frac{13 \zeta (3)}{32 g_\text{YM}^3}-\frac{13
   g_\text{YM}}{1536}\right]\\
   &+\frac{1}{N^{\frac52}}\left[\frac{1575 \zeta (9)}{16 g_\text{YM}^9}+\frac{g_\text{YM}^7}{6881280}-\frac{75 \zeta (5)}{64 g_\text{YM}^5}-\frac{5
   g_\text{YM}^3}{73728}\right]\\
   &+\frac{1}{N^{\frac72}}\biggl[\frac{2480625 \zeta (11)}{2048 g_\text{YM}^{11}}+\frac{25 g_\text{YM}^9}{2491416576}-\frac{80325 \zeta (7)}{8192 g_\text{YM}^7}-\frac{17 g_\text{YM}^5}{6291456}  \\
   &{}+\frac{1533
   \zeta (3)}{16384 g_\text{YM}^3}+\frac{511 g_\text{YM}}{262144}\biggr]
    +O(N^{-\frac92})\,.
}
These terms match the perturbative part of the Eisenstein series in \eqref{dlogZdm2More}, as defined in \eqref{EisensteinExpansion}. This match motivates the conjecture that the finite $g_\text{YM}$ expression for $\partial_m^2 \log Z \big\vert_{m=0}$ can be derived to any order in $1/N$ by computing the perturbative  terms as described in \cite{Chester:2019pvm}, and then simply replacing those by their Eisenstein completions using \eqref{EisensteinExpansion}.

Further evidence for \eqref{dlogZdm2More} comes from considering the instanton terms $ \partial_m^2 \log Z\big\vert_{m=0} ^\text{inst}$, which are written as expectation values of sums and products of eigenvalues. In the previous sections, these quantities were computed using the saddle-point expansion, which is valid to leading order in $1/N^2$ (including the subleading in $1/N$ term). Subleading corrections in $1/N^2$ can be computed using topological recursion \cite{Eynard:2004mh,Eynard:2008we}. This method naturally applies to the resolvent $W(y_1,\dots, y_n)$, which is defined as the connected expectation value
\es{W}{
W^n(y_1,\dots, y_n)\equiv N^{n-2} \left \langle \sum_{i_1} \frac{1}{y_1 -a_{i_1} }\cdots \sum_{i_n} \frac{1}{y_n -a_{i_n} } \right \rangle_\text{conn.}\,,
}
with the $1/N^2$ expansion 
\es{W2}{
W^n(y_1,\dots, y_n)\equiv\sum_{m=0}^\infty\frac{1}{N^{2m}} W^n_m(y_1,\dots, y_n)\,.
}
The coefficients $W^n_m$ can be computed for finite $\lambda$ for any $n,m$ in a Gaussian matrix model using a recursion formula in $n,m$, starting with the base case $W^1_0$, as reviewed for the Gaussian $U(N)$ SYM matrix model in \cite{Chester:2019pvm}. (See Footnote~\ref{UNFootnote}.) Topological recursion can then be applied to any expectation value that can be written in terms of the resolvents $W^n(y_1,\dots, y_n)$, for instance by taking derivatives or integrals in terms of $y_i$. Unfortunately, the operators that appear in the instanton terms \eqref{eq:pq-tab} are written as products over a restricted set of eigenvalues $a_i$, which cannot be easily related to $W^n(y_1,\dots, y_n)$. However, by expanding these products for small $a_i$, which is equivalent to a small $g_\text{YM}$ expansion, they can be expressed as an infinite sum of polynomials in $a_i$, whose expectation values can then be easily related to $W^n(y_1,\dots, y_n)$.

Let us begin by discussing the one-instanton case. By explicitly performing the sums and products in $I_1$  (Eq.~\eqref{oneInst}) for many small values of $N$, we find that $I_1$ can be expanded for small $a_i$ as
 \es{expansion}{
  I_1 &= - \frac{4 \Gamma(N + \frac 12)}{\sqrt{\pi} \Gamma(N)} - \frac{3 \Gamma(N - \frac 12)}{2\sqrt{\pi} \Gamma(N + 2)}C_2
   + \frac{315 \Gamma(N - \frac 32)}{32 \sqrt{\pi} \Gamma(N+4)} C_2^2 -\frac{15 (3 - N + 4 N^2) \Gamma(N - \frac 32)}{16 \sqrt{\pi} \Gamma(N+4)} C_4  + \cdots 
 }
where we defined the invariants
 \es{invars}{
  C_p = \sum_{j, k} (a_j - a_k)^p  \,.
   }
The expectation values of these $C_p$ can be related to coefficients of the large $y_i$ expansion of $n$-body resolvents $W^n(y_1,\dots, y_n)$ with $n\leq p$. Since the $C_p$ are degree $p$ polynomials in $a_i$, they must be proportional to $\lambda^{p/2}$ and their $1/N^2$ expansion truncates. For instance, for the $C_p$ that are shown in \eqref{expansion}, using the explicit expressions for $W_m^n$ in Appendix B of \cite{Chester:2019pvm} we find
 \es{Expectations}{
  \langle C_2 \rangle &= \lambda\left[\frac{N^2 }{8 \pi^2}-\frac{1}{8\pi^2}\right] \,,\quad  \langle C_2^2 \rangle =  \lambda^2\left[\frac{N^4 }{64 \pi^4}-\frac{1}{64\pi^4}\right]\,,\quad  \langle C_4 \rangle =  5\lambda^2\left[\frac{N^2 }{128 \pi^4}-\frac{1}{128\pi^4}\right] \,. \\
 }
We can then insert these expressions into the expectation value \eqref{expansion}, set $\lambda=g_\text{YM}^2N$, and expand in $1/N$ to get
\es{ZRatio22}{
\langle I_1\rangle&=\sqrt{N} \left[-\frac{4}{\sqrt{\pi }}-\frac{3 g_\text{YM}^2}{16 \pi ^{5/2}}+\frac{15 g_\text{YM}^4}{2048 \pi
   ^{9/2}}+O\left(g_\text{YM}^6\right)\right]\\
   &+\frac{1}{\sqrt{N}} \left[\frac{1}{2 \sqrt{\pi }}+\frac{15 g_\text{YM}^2}{128
   \pi ^{5/2}}+\frac{105 g_\text{YM}^4}{16384 \pi ^{9/2}}+O\left(g_\text{YM}^6\right)\right]\\
   &+\frac{1}{N^{\frac32}}
   \left[-\frac{1}{32 \sqrt{\pi }}+\frac{69 g_\text{YM}^2}{2048 \pi ^{5/2}}+\frac{2175 g_\text{YM}^4}{262144 \pi
   ^{9/2}}+O\left(g_\text{YM}^6\right)\right]\\
   &+\frac{1}{N^{\frac52}} \left[-\frac{5}{256 \sqrt{\pi
   }}+\frac{285 g_\text{YM}^2}{16384 \pi ^{5/2}}+\frac{24675 g_\text{YM}^4}{2097152 \pi
   ^{9/2}}+O\left(g_\text{YM}^6\right)\right]\\
   &+\frac{1}{N^{\frac72}} \left[\frac{21}{8192 \sqrt{\pi
   }}+\frac{5103 g_\text{YM}^2}{524288 \pi ^{5/2}}+\frac{1158885 g_\text{YM}^4}{67108864 \pi
   ^{9/2}}+O\left(g_\text{YM}^6\right)\right]+O(N^{-\frac92})\,.
 }
 This is consistent according to \eqref{oneInst} with the small $g_\text{YM}$ expansion of
\es{Expectation1}{
  &\left \langle  \partial_m^2 Z_\text{inst}^{(1)} (m, a_{ij}) \right \rangle \bigg|_{m=0} =e^{ \frac{8 \pi^2 }{g_\text{YM}^2} }\Big[ - \sqrt{N} \frac{16 K_1 (8 \pi^2 / g_\text{YM}^2)}{g_\text{YM}}  +  \frac{2 K_2 (8 \pi^2 / g_\text{YM}^2)}{\sqrt{N}g_\text{YM}} \\
 &\quad+\frac{1}{32g_\text{YM}N^{\frac32}}\left[  -13 K_1 (8 \pi^2 / g_\text{YM}^2)+9 K_3 (8 \pi^2 / g_\text{YM}^2) \right] \\
 &\quad+\frac{1}{128g_\text{YM}N^{\frac52}}\left[  -25 K_2 (8 \pi^2 / g_\text{YM}^2)+15 K_4 (8 \pi^2 / g_\text{YM}^2) \right] \\
 &\quad +\frac{1}{g_\text{YM}N^{\frac72}}\left[\frac{1533 K_1\left(\frac{8 \pi ^2}{g_\text{YM}^2}\right)}{16384 }-\frac{5355 K_3\left(\frac{8 \pi ^2}{g_\text{YM}^2}\right)}{32768 }+\frac{2625
   K_5\left(\frac{8 \pi ^2}{g_\text{YM}^2}\right)}{32768 }\right]+ O(N^{-\frac92})\Big] \,,\\
 } 
 which generalizes \eqref{ZRatioFinal} to higher orders in $1/N$, and describes the one-instanton contribution to the Eisenstein series in \eqref{c2}. In Appendix~\ref{highInst}, we similarly match all instantons up to $k=12$ to order $O(g_\text{YM}^4)$.
 
 We should end by pointing out that the $SL(2, \Z)$-invariant expression \eqref{dlogZdm2More2} has the property that the coefficients multiplying $\pi^{-r} E(r, \tau, \bar \tau)$ are all rational numbers.  It would be interesting to understand the significance of this fact in relation to the set of modular functions that appear in the expression for the $\cN = 4$ SYM correlator and superstring scattering amplitudes.

\section{Conclusion}
\label{DISCUSSION}
 
In this paper, we studied the four-point correlator $\langle SSSS \rangle$ of the superconformal primary operator $S$ transforming in the ${\bf 20}'$ of $SO(6)_R$ in the $\cN=4$ SYM theory, in the ``very strong coupling'' limit in which $N$ is sent to infinity at fixed $g_\text{YM}$. In this limit, the action of $SL(2, \Z)$ modular transformations on the $\langle SSSS \rangle$ correlator is manifest.  In particular, we studied the constraints on $\langle SSSS \rangle$ coming from the flat space limit of the IIB string theory amplitudes, and those coming from the integrated four-point function $\tau_2^2\partial_\tau\partial_{\bar\tau}\partial_m^2\log Z\big\vert_{m=0}$.  The latter can be computed using supersymmetric localization.   Starting from Pestun's localization expression \cite{Pestun:2007rz} for the partition function $Z$, we argued that when $\tau_2^2\partial_\tau\partial_{\bar\tau}\partial_m^2\log Z\big\vert_{m=0}$ is expanded in $1/N$, the first two sub-leading terms (of orders $N^\frac12$ and $N^{-\frac12}$, respectively) can be written as the Eisenstein series $E( \threeh,\tau,\bar\tau)$ and $E( \fiveh,\tau,\bar\tau)$, respectively.  Our argument is not completely rigorous because it relies on studying the $k$-instanton contribution for many values of $k$ and deducing the general pattern, but we hope that it should be possible to provide a more rigorous argument in future work.   Using solely the relation between the integrated $\langle SSSS\rangle$ correlator and $\tau_2^2\partial_\tau\partial_{\bar\tau}\partial_m^2\log Z\big\vert_{m=0}$ from \cite{Binder:2019jwn}, we completely determined the $N^{\frac12}$ term in the large $N$ expansion of $\langle SSSS \rangle$.  This term corresponds to an effective $\ell_s^{-2} R^4$ coupling in $AdS_5$, which, in the flat space limit, matches the $\ell_s^{-2} R^4$ contribution to the Type IIB graviton S-matrix as computed at finite string coupling $g_s$ in \cite{Green:1997tv,Green:1997as,Green:1998by}. This is a precision test of AdS/CFT at finite $g_s$!  We then used the $\ell_s^2 D^4R^4$ term in the Type IIB S-matrix, which is also known at finite $g_s$, as well as the $N^{-\frac12}$ term in $\tau_2^2\partial_\tau\partial_{\bar\tau}\partial_m^2\log Z\big\vert_{m=0}$ to completely determine $\langle SSSS\rangle$ at order $N^{-\frac12}$\@.  In Mellin space, this expression contains two polynomial terms, both proportional to $E( \fiveh,\tau,\bar\tau)$, one corresponding to an $\ell_s^2 D^4 R^4$ contact term in $AdS_5$ and one to an $\frac{\ell_s^2}{L^4}  R^4$ term. Finally, using a small $g_\text{YM}$ expansion, we gave non-trivial evidence that each of the terms in the $1/N$ expansion of $\tau_2^2\partial_\tau\partial_{\bar\tau}\partial_m^2\log Z\big\vert_{m=0}$ is a finite linear combination of non-holomorphic Eisenstein series.

The fact that we can derive the full CFT correlator at order $N^\half$ generalizes the analysis in   \cite{Dorey:1998xe,Dorey:1999pd}  of the $k$-instanton measure in  the large-$N$ limit of $SU(N)$ $\cN=4$ SYM theory at lowest order in $g_{\text{YM}}$.  That analysis, which was based on a large-$N$ saddle point solution of the ADHM constraints, demonstrated a number of general features of the holographic relationship between Yang-Mills instantons and D-instantons.
 In particular, in the large-$N$ saddle point approximation the $k$-instanton moduli space (where $k\ll N$)   is dominated by the region in which the positions of the instantons and their scale sizes coincide, so they are represented by a single point in $AdS_5\times S^5$,  with all the instantons in commuting  $SU(2)$ subgroups of $SU(N)$.
This is interpreted holographically as a single D-instanton of charge $k$ in the dual type IIB theory. Furthermore, the divisor sum $\sigma_{-2}(|k|)$ arises as the partition function of the D-instanton matrix model, i.e.~the partition function of 10d $\cN=1$ $SU(k)$ Yang--Mills theory reduced to zero space-time dimensions.   
 
 The result of this large-$N$ ADHM analysis is reproduced in our procedure by the first term in the small-$g_{\text{YM}}$ expansion of the  Bessel function in the $k$th Fourier mode of the Eisenstein series $E(\threeh,\tau, \bar \tau)$ (the function $\cF_k$ defined in \eqref{nonzeroeisen}).  The fact that the dominant contribution to the  Nekrasov partition function in the $m\to 0$ limit has a single cluster of boxes should correspond to properties of the large-$N$ analysis of the ADHM construction.  However, this correspondence is difficult to make precise since our analysis is based on taking a limit of the non-conformal $\cN=2^*$ whereas conformal invariance is explicit in the large-$N$ ADHM construction.   The connection of the D-instanton measure with the $SU(k)$ D-instanton matrix model partition function is also not obvious in our procedure.  Nevertheless, the fact that our procedure packages an infinite number of perturbative corrections to the $k$-instanton contribution into a $K$-Bessel function is a most significant generalization of \cite{Dorey:1998xe,Dorey:1999pd}  and an essential requirement of $SL(2,\Z)$ covariance.

As shown in Section~\ref{MANYEISENSTEINSERIES} our integrated constraint $\tau_2^2 \partial_\tau\partial_{\bar  \tau} \partial_m^2 \log Z\big\vert_{m=0}$ has an expansion in half-integer powers of $1/N$ (apart from the first term).   However, it is well known that the low energy expansion of the string amplitude does also contain even powers of $\ell_s^2$, which lead to integer powers of $1/N$, the most relevant one being the 1/8-BPS interaction $D^6 R^4$. 
  A more complete analysis of the holographic correspondence should therefore also include terms with integer powers of $1/N$.  We expect that such terms will appear in other quantities that can be computed using supersymmetric localization, such as $\partial_m^4 \log Z\big\vert_{m=0}$ or $\partial_b^2 \partial_m^2 \log Z\big\vert_{m=0, b=1}$, where  $b=\epsilon_1/\epsilon_2$ is a parameter that  defines the squashing deformation of $S^4$ that appear in \eqref{Ik-gen} (recall that up to now $\epsilon_1=\epsilon_2=1$).   We expect that determining these three distinct integrated four-point correlation  functions should eliminate the ambiguities in determining the expansion of the $AdS_5\times S^5$ type IIB string theory  amplitudes up to order $D^6R^4$.  In other words, this procedure should uniquely  determine the BPS protected interactions without the need to input known results from flat-space type IIB superstring theory.

We have so far only considered four-point correlators.  It was argued in \cite{Intriligator:1998ig,Intriligator:1999ff} that all four-point functions of short operators are invariant under both the bonus $U(1)$ symmetry as well as $SL(2, \Z)$. As we have seen,  the four-point integrated correlators are linear combinations of non-holomorphic Eisenstein series, which are manifestly $SL(2, \Z)$ invariant. Such a statement would not hold for $(n \geq 5)$-point functions, which transform as modular forms with non-trivial modular weights.\footnote{
A  modular form with  holomorphic and anti-holomorphic weights $(w,w')$ transforms as 
$f^{(w,w')} (\tau' , \bar\tau')\to (c\tau+d)^w(c\bar\tau+d)^{w}    f^{(w,w')}(\tau,\bar\tau) $
under the action of $SL(2,\Z)$. } These correlation functions should correspond to higher-point superstring amplitudes that violate the $U(1)$ R-symmetry of type IIB supergravity. Such  $U(1)$-violating superstring amplitudes (especially those that violate $U(1)$ maximally) are identified in \cite{Boels:2012zr, Green:2019rhz}, and more importantly the F-terms (terms up to the same number of derivatives as $D^6R^4$) have also been determined using maximal supersymmetry and $SL(2, \Z)$ symmetry \cite{Green:2019rhz}.  The coefficients of these $U(1)$-violating interactions are modular forms with non-zero modular weights,  and it would be of interest to understand how they arise from the supersymmetric localization computation. 

Lastly, it is interesting to compare the calculation presented in this paper to calculations done in the 3d ABJM theory \cite{Aharony:2008ug} with gauge group $U(N)_k \times U(N)_{-k}$ and ${\cal N} = 6$ supersymmetry.  An analogous computation that includes non-perturbative contributions can also be performed in that case \cite{Binder:2019mpb}, in the very strong coupling limit in which $N$ is taken to infinity while $N / k^5$ is held fixed.  In this limit, the ABJM theory is dual to type IIA string theory on $AdS_4 \times \CP^3$ at finite string coupling $g_s$.   However, in this case, all the non-perturbative contributions to the type IIA scattering amplitudes of the lowest closed string states vanish.  Thus, in order to obtain non-trivial non-perturbative contributions to CFT correlators that can be matched to string scattering amplitudes, one is led to consider the case of the 4d ${\cal N} = 4$ theory that was studied in the present paper.  Nevertheless, it is worth pointing out that the closest 3d analog of the formulas presented in Section~\ref{MANYEISENSTEINSERIES}  that include resummed instanton contributions would be the mass-deformed partition function of ABJM theory that can be computed \cite{Nosaka:2015iiw} to all orders in the $1/N$ expansion using the Fermi gas method developed in \cite{Marino:2011eh}.

\section*{Acknowledgments}

We thank Ofer Aharony, Massimo Bianchi, Nick Dorey, Francesco Fucito, Francisco Morales, Rodolfo Russo, Yuji Tachikawa,  and  Xi Yin for useful discussions.  SMC is supported in part by a Zuckerman STEM Leadership Fellowship.   MBG has been partially supported by STFC consolidated grant ST/L000385/1 and by a Leverhulme Emeritus Fellowship.   The work of SSP was supported in part by the US NSF under Grant No.~PHY-1820651 and by the Simons Foundation Grant No.~488653.  The work of YW was supported in part by the US NSF under Grant No.\ PHY-1620059 and by the Simons Foundation Grant No.\ 488653. CW is supported by a Royal Society University Research Fellowship No. UF160350.  We would also like to thank the organizers of ``Scattering amplitudes and the conformal bootstrap'' workshop and the Aspen Center for Physics (ACP) for hospitality during the initial stages of this work.   The ACP is supported by National Science Foundation Grant No.~PHY-1607611.   

\appendix

\section{Non-holomorphic Eisenstein series}
\label{app:Einstein}

The first four terms in the low energy expansion of the four-graviton amplitude in the type IIB superstring theory correspond to BPS protected effective interactions that take the form 
\es{fournew}{
 {\cal A} (\Mands,\Mandt)  = & \frac{R^4}{\ell_s^8 g_s^2} \, \left(\frac{1}{\Mands \Mandt \Mandu} + \frac{\ell_s^6 g_s^{\frac 32}}{2^6} \, E(\threeh,\tau,\bar\tau)+  \frac{\ell_s^{10}  g_s^{\frac 52}}{2^{11} }\, (\Mands^2+\Mandt^2+\Mandu^2) E(\fiveh,\tau,\bar\tau) \right.\\
&\left.\qquad+ \frac{\ell_s^{12} g_s^3}{2^{12}}  (\Mands^3+\Mandt^3+\Mandu^3) \cE(\tau,\bar\tau) + \dots \right)\,,}
where $R$ signifies the linearised Weyl curvature tensor, which has the form $R_{\mu\nu\rho\sigma}= k_{[\mu}\epsilon_{\nu]} k_{[\rho} {\epsilon}_{\sigma]}$ (where $[\cdots]$ denotes anti-symmetrization of the indices),  where $k_\mu$ is a null ten-dimensional momentum and  $\epsilon_{\nu}  {\epsilon}_{\sigma}$ is a graviton polarization. The symbol  $R^4$ denotes the particular contraction of four curvature tensors that is implied by ten-dimensional  $\cN=2$  supersymmetry.

The function $E(r,\tau,\bar\tau)$ is a non-holomorphic Eisenstein series, which is a modular function of $(\tau, \bar \tau)$ that satisfies the Laplace eigenvalue equation
\es{lapdef}{\left(\Delta_\tau - r(r-1)\right) E(r,\tau,\bar\tau)=0\,,
} 
where the hyperbolic laplacian is defined by  $\Delta_\tau= 4\tau_2^2\partial_\tau\partial_{\bar\tau} = \tau_2^2\, (\partial_{\tau_1}^2+\partial_{\tau_2}^2)$.
Assuming moderate growth as $\tau_2\to \infty$ (in other words, assuming $E(r,\tau,\bar\tau)$ grows no faster than $\tau_2^a$ for any $a$,  so it is no more singular than perturbative string theory in the $g_s\to 0$ limit) the unique $SL(2,\Z)$-invariant solution to this equation is\footnote{It is often convenient to define the non-holomorphic Eisenstein series in terms of a Poincar\'e series  $\hat E(r,\tau,\bar\tau) =\half \pi^{-r}\,\Gamma(r) E(r,\tau,\bar\tau)$, where  
$$\hat E(r,\tau,\bar\tau) =\pi^{-r}\Gamma(r) \zeta(2r)\frac{1}{2} \sum_{\gamma\in \Gamma_\infty\backslash SL(2,\Z)} \Im(\gamma(\tau))^r,$$
and
$ \Gamma_\infty=\left\{\begin{pmatrix} 1&n \cr 0&1\cr  \end{pmatrix}  \Big|  \, n  \in \Z \right\}.$ The expression $\hat E(r,\tau,\bar\tau)$ is manifestly invariant under $SL(2,\Z)$ and satisfies the important functional relation $\hat  E(r,\tau,\bar\tau)=\hat E(1-r,\tau,\bar\tau)$. }
\es{eisensol}{
E(r,\tau,\bar\tau) = \sum_{(m,n)\,\neq\,(0,0)}\frac{\tau_2^r}{|m\tau+n|^{2r}}    \,.}
It is straightforward to show that $E(r,\tau,\bar\tau)$ is invariant under a $SL(2,\Z)$ transformation, 
\ie
E(r,\tau,\bar\tau) \to E(r,\tau',\bar\tau') \, , \quad {\rm with} \quad \tau \rightarrow \tau' = \frac{a \tau+b }{c \tau+d} \, ,
\fe
for $a, b, c,d \in \Z$ and $ad-bc=1$. 

A non-holomorphic Eisenstein series has an expansion in Fourier modes of the form
\es{eisenfourier}{
E(r,\tau,\bar\tau)    =   \sum_{k\in\Z} \cF_k(r,\tau_2) \, e^{2\pi i k \tau_1}\,,}
where the zero mode consists of two power behaved terms,
\es{eisenzero}{
\cF_0 (r,\tau_2)  = 2\zeta(2r)\,  \tau_2^r \  +  \ \frac{2\sqrt \pi \,\Gamma(r-\frac{1}{2}) \zeta(2r-1)}{\Gamma(r)}\, \tau_2^{1-r} \,,}
and the non-zero modes  are proportional to $K$-Bessel functions,
\es{nonzeroeisen}{
\cF_k(r,\tau_2)  =   \frac{2\,\pi^r}{\Gamma(r)}\,  |k|^{r-\half} \, \sigma_{1-2r}(|k|)
\sqrt{\tau_2}\,K_{r-\half}(2\pi |k|\tau_2) \,, \qquad  k\neq 0\,,}
where the divisor sum is defined by  
\es{divisorsum}{
\sigma_p(k)=\sum_{d>0,{d|k}}  d^p} 
for $k>0$, and $\sigma_{-p}(k) = k^{-p}\, \sigma_p(k)$. 

The two power-behaved terms  in $\cF_0(r,\tau_2)$  in \eqref{eisenzero}  correspond to tree-level and $(r-\half)$-loop contributions in string perturbation theory.  Using the asymptotic behavior of the $K$-Bessel function
\es{Kasymptotic}{
K_\nu(z) \sim \sqrt{\frac{\pi}{2z}}   \, e^{-z} \left(1+  O\left( \frac{1}{z} \right)    \right) \,, }
we see that the non-zero mode $\cF_k(r,\tau_2)$ behaves as $e^{-2\pi |k| \tau_2}$ and has the form of a $k$ D-instanton contribution.

The terms proportional to $E(\threeh,\tau,\bar\tau)$ and $E(\fiveh,\tau,\bar\tau)$ in \eqref{fournew}  are coefficients of $R^4$ and $D^4R^4$ interactions in the type IIB low energy effective action.  These are, respectively, $1/2$-BPS and $1/4$-BPS interactions.  The last term in  \eqref{fournew} corresponds to  a term proportional to the $1/8$-BPS interaction $D^6R^4$,  with a coefficient $\cE(\tau,\bar\tau)$ that satisfies the inhomogeneous Laplace eigenvalue equation \cite{Green:2005ba,Wang:2015jna}
\es{d6r4}{(\Delta_\tau-12)\, \cE(\tau,\bar\tau)=-E^2(\threeh,\tau,\bar\tau)\,.}
The solution to this equation  \cite{Green:2014yxa}  is qualitatively different from an Eisenstein series.  The zero mode is of the form in the large-$\tau_2$ limit,
\es{d6r4sol}{\cF_{0,\cE}(\tau_2):= \int_{-\half}^\half d\tau_1 \cE(\tau,\bar\tau) =   {2\,\zeta(3)^2\,\over3}\,\tau_2^3 \ + \ {4\,\zeta(2)\,\zeta(3)\over3}\,\tau_2 \
  + \ {4\,\zeta(4)\over \tau_2} +O(e^{-4\pi\tau_2})\,.}
The  power-behaved contributions correspond to  string perturbation theory up to genus three.  The symbol $O(e^{-4\pi |k|\tau_2})$ denotes a specific infinite sum of D-instanton--anti D-instanton contributions with zero total instanton number (details of which are in   \cite{Green:2014yxa}).   Similarly,  each mode of non-zero mode number $k$  has the form of a sum of D-instanton--anti D-instanton contributions with instanton numbers $k_1$ and $k_2$ satisfying $k_1+k_2=k\ne0$.    We see in the main text that such a  $(\Mands^3+\Mandt^3+\Mandu^3) \, R^4$   contribution  does not arise from our analysis of the flat-space limit of  $\tau_2^2 \partial_\tau  \partial_{\bar \tau}   \partial^2_{m}\, \log Z$ since its contribution to the integrated correlation  function vanishes.

The four terms in the low energy expansion of the four-graviton amplitude explicitly shown in \eqref{fournew} correspond to local BPS interactions that are fully determined by supersymmetry, while higher derivative terms are not expected to be protected and have not been fully determined.

\section{Rectangular dominance}
\label{app:rectd} 

As discussed in the main text, the full Nekrasov partition function for the mass-deformed $\cN=2^*$ $SU(N)$ SYM with squashing  parameters $\epsilon_{1,2}$ at instanton number $k$ is given by a sum over $N$-tuples of Young diagrams $\vec Y=(Y_1, Y_2, \ldots, Y_N)$ with $k$ boxes in total,
\es{genNek}{
	& Z^{(k)}_{\rm inst}(m,a_{ij},\epsilon_{1,2})= 
	\sum_{|\vec Y|=k} Z_{\vec Y}
	\\ 
	&	Z_{\vec{Y}}  \equiv 
	\frac{ \prod_{i,j=1}^N	 
		\prod_{s\in Y_i} (E(a_{ij},Y_i,Y_j,s)-i m -\epsilon_+/2)\prod_{t\in Y_j} (-E(a_{ji},Y_j,Y_i,t)-i m +\epsilon_+/2)
	}{
	\prod_{s\in Y_i}  E(a_{ij},Y_i,Y_j,s) \prod_{t\in Y_j} (\epsilon_+-E(a_{ji},Y_j,Y_i,t)) } \,,
}
where $\epsilon_\pm \equiv \epsilon_1\pm \epsilon_2$  as in the main text, and
\ie
E(a_{ij},Y_i,Y_j,s) \equiv i a_{ji}-\epsilon_1 h_j(s)+\epsilon_2 (v_i(s)+1) \,.
\fe
Here, $s$ labels a box $(\A,\B)$ ($\A$-th column and $\B$-th row) in a given Young diagram as Figure~\ref{fig:YTcoords}, and $h_i(s)$ and $v_i(s)$ denote the arm-length and leg-length, respectively, of the box $s$ in the diagram $Y_i$.  (Each individual Young diagram $Y$ consists of columns of non-increasing heights $\lambda_1\geq \lambda_2 \geq \dots \geq \lambda_l$ with $\lambda_l\geq 1$. The transpose (conjugate) diagram $Y^T$ has columns of heights $\lambda^T_1\geq \lambda^T_2\geq \dots\geq \lambda^T_m$ with $\lambda^T_m\geq 1$.
Then the arm-length $h$ and leg-length $v$  of the box $s$ in $Y$ are given by (see Figure~\ref{fig:YTcoords})
\ie
h(s)=\lambda^T_\B-\A \,, \qquad v(s)=\lambda_\A-\B.
\fe
Note that the definitions of $h$ and $v$ extend  beyond boxes in $Y$ to the entire quadrant $(\A,\B)\in \mZ_+^2$ in the obvious way. In particular they can be negative (e.g.~when $Y$ is empty).)

In the rest of this appendix, we prove the following theorem:
\begin{theorem}
	For the second mass derivative of the instanton partition function at $m=0$
	\ie
	I_{k } (a_{ij}) \equiv  \partial_m^2 \left( \lim_{\ep_{1,2}\to 1} Z^{(k)}_{\rm inst}(m,a_{ij},\epsilon_{1,2})\right)   \bigg|_{m=0}\,  
	\label{m2nek}
	\fe
	only the summands in \eqref{genNek} with $\vec Y$ of the following type
	\ie
	Y_i=&\varnothing~{\rm if}~i\neq \hat i 
	\\
	Y_{\hat i}=&\underbrace{[p,p,\dots,p]}_{q}~{\rm with~}pq=k~{\rm for~}p,q\in \mZ_+,~
	\fe
	and those with $Y_{\hat i}$ replaced by its {\it partial transpositions} (which we define next) contribute. 
	\label{thmvan}
\end{theorem}  

Given a Young diagram $Y=[\lambda_1,\lambda_2,\dots,\lambda_l]$ (with $\lambda_l\geq 1$), we define its \textit{partial transposition} at position $\A$, with $1\leq \A\leq l$, to be 
\ie
{\rm PT}_\A (Y)=[\lambda_1,\lambda_2,\dots,\lambda_{\A-1},\lambda'_1,\dots,\lambda'_{\lambda_\A}] \,,
\fe
where $\lambda'_\B$ are column lengths of the transpose of the subdiagram $P=[\lambda_\A,\lambda_{\A+1},\dots,\lambda_l ]$, if $\lambda_{\A-1}\geq \lambda_1'$. Equivalently the partial transposition at $\A$-th column is given by the disjoint union
\ie
{\rm PT}_\A (Y)=(Y\backslash P) \sqcup P^T \,,
\fe
where $P$ denotes the block (Young subdiagram) to the right of the $(\A-1)$-st column, and $Y\backslash P$ is the complement Young subdiagram (see examples in Figure~\ref{fig:pt}).
In particular, the usual transposition is a partial transposition at $\A=1$. For notational simplicity, we will often suppress the subscript $\A$ when the context is clear.

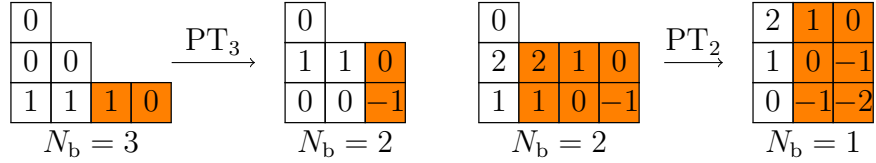
\begin{figure}[htb!]
	\centering
	\begin{tikzpicture}[inner sep=0in,outer sep=0in]
	\node (1a)  {
		\begin{ytableau}
		0    \\
		0 &  0 \\
		1 &  1  & *(orange) 1 & *(orange) 0\\
		\end{ytableau}
	};
	\node (2a) [right=1.5 of 1a]  {
		\begin{ytableau}
		0    \\
		1 &  1  & *(orange) 0 \\
		0 &  0  & *(orange) -1  \\
		\end{ytableau}
	};
	\draw[->] (1a.east) -- ($(2a)+(-1.2,0)$) node[midway, above=.1] {PT$_3$};
	\node (1b)  [below=.1 of 1a]{$N_{\rm b}=3$};
	\node (2b)  [below=.1 of 2a]{$N_{\rm b}=2$};
	\end{tikzpicture}
	\quad~~
	\begin{tikzpicture}[inner sep=0in,outer sep=0in]
	\node (1a)  {
		\begin{ytableau}
		0    \\
		2 &  *(orange)2 &  *(orange)1&  *(orange)0\\
		1 & *(orange) 1  & *(orange) 0 & *(orange) -1\\
		\end{ytableau}
	};
	\node (2a) [right=1.5 of 1a]  {
		\begin{ytableau}
		2  &  *(orange)1&  *(orange)0  \\
		1 &  *(orange)0 &  *(orange)-1\\
		0 & *(orange) -1  & *(orange) -2 \\
		\end{ytableau}
	};
	\draw[->] ($(1a)+(1.4,0)$) -- ($(2a)+(-1.2,0)$) node[midway, above=.1] {PT$_2$};
	\node (1b)  [below=.1 of 1a]{$N_{\rm b}=2$};
	\node (2b)  [below=.1 of 2a]{$N_{\rm b}=1$};
	\end{tikzpicture}
	\caption{
		Example of partial transpositions and the corresponding changes in $N_{\rm b}(Y)$. The subdiagram $P$ is colored in orange. 
		The numbers in the individual boxes $s$ is the value $d(s)$. 
	}
	\label{fig:pt}
\end{figure}

We will need the following useful properties of the map $\rm PT_\A(\cdot)$. 
  We start with the obvious lemma:
\begin{lemma}
	For a Young diagram $Y$, the multiset $\Delta(Y)=\{\A+\B\,|\,(\A,\B)\in Y\}$
	is invariant under partial transpositions. 
	\label{lemmainvset}
\end{lemma}

Consequently the partial transpositions preserve the set of poles \eqref{YTpoles2} in the contour integral for the instanton partition function. A useful corollary that follows  is the following:
\begin{corollary}
	Among all the Young diagrams related by (a sequence of) partial transpositions to a rectangular diagram $Y_{p\times q}$ with $p$ columns of height $q$ each, the maximum width (or height) is  $p+q-1$.
	\label{coro:Ysize}
\end{corollary}
\begin{proof}
	A single partial transposition at the right-most column gives a Young diagram with width $p+q-1$. The fact that this is the maximal value that can be achieved by any sequence of partial transpositions follows from
	the previous lemma by noting that 
	\ie
	\max [ \Delta(Y_{p\times q})]=p+q.
	\fe
\end{proof}

\begin{lemma}
	For Young diagram $Y=[\lambda_1,\dots,\lambda_l]$ with $\lambda_l\geq 1$, the  multiset
	\ie
	\Delta_B(Y)\equiv \{\A+\lambda_\A\,|\,1\leq \A \leq M\} 
	\fe
	is invariant under  partial transpositions, where $M$ is any integer equal to or larger than the maximum width of Young diagrams related to $Y$ by partial transpositions. Note that $\lambda_\A$ for $\A>l$ is defined to be 0.
	\label{lemmadiaginv}
\end{lemma}
\begin{proof}
	It suffices to prove that $\Delta_B(Y)=\Delta_B(Y^T)$, because for any partial transposition with respect to a subdiagram $P\subset Y$, $\Delta_B(Y\backslash P)$ is clearly invariant.  The entries in $\Delta_B(Y)$ consists of $\A+\B$ for boxes $(\A,\B)$ located on the North-East  boundary of the Young diagram.  
	
		We proceed by induction. If $Y$ consists of a single box, $\Delta_B(Y)=\Delta_B(Y^T)$ trivially for any $M\geq 1$. Suppose $\Delta_B(Y)=\Delta_B(Y^T)$ for some large enough $M$ as specified in the statement of the lemma. We show it continues to hold upon adding a box to $Y$. First of all, in order to generate a Young diagram, the additional box must be added to an \textit{inward} corner of the North-East border of $Y$, i.e.~a coordinate $(\A,\B)$ that satisfies
	\ie
	(\A,\B)=(\A,\lambda_{\A}+1)=(\lambda_{\B}^T+1,\B).
	\label{indh}
	\fe
We denote the new Young diagram obtained in this fashion by $Y'$.  Then $\Delta_B(Y')$  is given by $\Delta_B(Y)$ with the element $\A+\lambda_\A$ (an orange entry in Figure~\ref{fig:slidediag}) replaced by $\A+\lambda_\A+1$. Similarly $\Delta_B(Y'^T)$  is given by $\Delta_B(Y^T)=\Delta_B(Y)$ with the element $\lambda_{\B}^T+\B$ replaced by $\lambda_{\B}^T+\B+1$.  Due to \eqref{indh}, we have $\A+\lambda_\A = \lambda_{\B}^T+\B = \A + \B  -1$, as well as $\A+\lambda_\A+1 = \lambda_{\B}^T+\B+1 = \A + \B$.  Thus, $\Delta_B(Y')$ and $\Delta_B(Y'^T)$ are each obtained from  $\Delta_B(Y)$ and $\Delta_B(Y^T)$, respectively, by replacing an occurrence of $\A + \B - 1$ with $\A + \B$.  Since by assumption $\Delta_B(Y) = \Delta_B(Y^T)$, we conclude that $\Delta_B(Y')=\Delta_B(Y'^T)$.  See Figure~\ref{fig:slidediag} for an example.   
\end{proof}
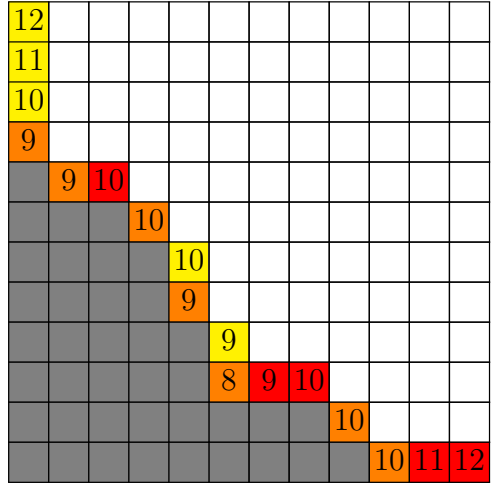
\begin{figure}[htb!]
	\centering
	\begin{tikzpicture}[inner sep=0in,outer sep=0in]
	\node (1a)  {
		\begin{ytableau}
		 *(yellow) 12 &   & &   & &  & &   & & & &  \\
		 *(yellow) 11 &   & &   & &  & &   &  & & &   \\
		 *(yellow) 10 &   & &   & &  & &   &  & & &   \\
		  *(orange) 9 &   & &   & &  & &   &  & & &   \\
	*(gray)	&  *(orange) 9  &  *(red) 10 &   & &  & &   &  & & &   \\
		*(gray)  & *(gray) & *(gray) & *(orange) 10  & &  & &   &  & & &   \\ 
		*(gray) & *(gray) & *(gray) & *(gray)  & *(yellow) 10 &  & &   &  & & &   \\
		*(gray) & *(gray) & *(gray) & *(gray)    & *(orange) 9 &  & &   &  & & &   \\
		*(gray) & *(gray) & *(gray) &  *(gray) &  *(gray) & *(yellow) 9 &   &  &  & & &   \\
		*(gray) & *(gray) & *(gray) &  *(gray)  &*(gray)   & *(orange) 8  & *(red) 9  & *(red) 10  &  & & &   \\
		*(gray) & *(gray) & *(gray) &  *(gray)   & *(gray)   &  *(gray)  & *(gray) & *(gray)   & 
		*(orange) 10   & & &   \\
		*(gray) & *(gray) & *(gray) &  *(gray)   & 	*(gray) & *(gray) & *(gray) &  *(gray)   &  *(gray) & *(orange) 10  &  *(red) 11 &  *(red) 12 
		\end{ytableau}
	};
	\end{tikzpicture}
	\caption{
		An example with $Y=[8,7,7,6,4,2,2,2,1]$ (gray boxes) in a square of size $M=12$.  The boxes at $(\lambda_\A^T + 1,\A)$ are colored in yellow; the boxes at $(\A,\lambda_\A + 1)$ are colored in red; the boxes common to both are colored in orange.  The numbers in the yellow and orange boxes are part of $\Delta_B(Y^T)$ while those in the red and orange boxes are part of $\Delta_B(Y)$.  Thus, in this example, $\Delta_B(Y) = \{9, 9, 10, 10, 9, 8, 9, 10, 10, 10, 11, 12 \}$ and $\Delta_B(Y^T) =\{ 10, 10, 8, 9, 9, 10, 10, 9, 9, 10, 11, 12 \}$.
	}
	\label{fig:slidediag}
\end{figure}

For a box with coordinates $s = (\A, \B)$ in a Young diagram $Y$, we define $d(s)=h(s)-v(s)$. Then we also 
define $n_a(Y)$ to be number of coordinates $s$ with $d(s)=a$. 

\begin{lemma}
	For a Young diagram $Y$, 
	\es{muDef}{
		\mu(Y)=2n_0(Y)- (n_1(Y)+n_{-1}(Y))
	}  is invariant under partial transpositions. In fact both $n_0(Y)$ and $n_1(Y)+n_{-1}(Y)$ are separately invariant.
	\label{lemmainv}
\end{lemma}
\begin{proof}
	Under a partial transposition that involves transposing the subdiagram $P$ of $Y$, the box at $(\A,\B) \in P$ gets mapped to the box at $(\B,\A) \in P^T$. Thus, focusing on boxes in $P$,  the boxes with $d(s)=0$ are preserved, whereas the boxes with $d(s)=\pm1$ are mapped to boxes with $d(s)=\mp 1$---See Figure~\ref{fig:pt}. Consequently, the contribution to $\mu(Y)$ from $P$ does not change under the transposition. Under the partial transposition, the values of $d(s)$ for the boxes in $s \in Y \backslash P$ get permuted within each column. To see this let's consider for example a partial transposition PT$_{\hat \A}$ with respect to the sub-diagram  $P=[\rho_1,\rho_2,\dots]$   (namely $\rho_a=\lambda_{\hat\A+a-1}$) with $\lambda_{\hat \A-1}\geq \rho^T_1$ (so that it's a nontrivial operation). Focusing on the boxes in the $\A$-th column with $1\leq \A \leq \hat{\A}-1$, their $d=h-v$ values before the partial transposition are given by
	\ie
	\{\rho^T_\B+\B+(\hat \A-\lambda_\A-\A)| 1\leq \B \leq \lambda_\A \}
	\fe 
 whereas	after the partial transposition, we have
		\ie
		\{\rho_\B+\B+(\hat \A-\lambda_\A-\A)| 1\leq \B \leq \lambda_\A \}.
		\fe
From Lemma~\ref{lemmadiaginv}, we see the above two multisets are the same.  Hence the contribution to $n_d$ coming from $Y\backslash P$ does not change under the transposition of $P$ (see examples in Figure~\ref{fig:pt}).   Therefore we conclude $\mu (Y)$ is a partial transposition invariant, and clearly so is $n_0(Y)$.
\end{proof}

In particular, a rectangular Young diagram $Y_{p\times q}$ with $p$ columns and $q$ rows has
 \es{nRect}{
  n_0(Y) = \min (p, q)\,, \qquad
   n_1(Y) = \min(p-1, q) \,, \qquad
    n_{-1}(Y) = \min(p, q-1) \,,
 }
which implies that 
\ie
\mu(Y_{p\times q}) =\begin{cases}
	2 & {\rm if}~p=q \,, 
	\\
	1 & {\rm if}~p\neq q \,.
\end{cases}
\fe

To proceed, let us make two more definitions.  We define $N_{\rm b}(Y) \geq 1$ to be the minimal number of rectangular blocks (in a horizontal decomposition) in  a Young diagram $Y$ (see Figure~\ref{fig:pt} for examples). Note that while $N_{\rm b}(Y)$ is invariant under (full) transposition, it typically changes under partial transpositions (see Figure~\ref{fig:pt} for example). We give the set of all Young diagram $Y$ with given size $|Y|=k$ a lexicographic (total) ordering as follows. For two Young diagrams $Y=[\lambda_1,\lambda_2,\dots, \lambda_l]$ and $Y'=[\rho_1,\rho_2,\dots,\rho_m]$, we have $Y \geq Y'$ when $\lambda_1>\rho_1$ or when $\lambda_1 = \rho_1$ and $\lambda_2>\rho_2$ or when $\lambda_1=\rho_1$, $\lambda_2 = \rho_2$, and $\lambda_3 > \rho_3$, etc., or when $\lambda_i = \rho_i$ for all $i$. 

Given a Young diagram $Y$, we consider all possible partial transpositions, and among the ones that minimize $N_{\rm b}(Y)$, we take the smallest Young diagram according to the lexicographic ordering, and call it $Y_{\rm min}$. For example if $Y$ is related by (possibly several successive) partial transpositions to a rectangular Young diagram $Y_{p\times q}$ for some $p\leq q$, then clearly
\ie
Y_{\rm min}=Y_{p\times q} \,,
\fe
because $N_\text{b}(Y_{p \times q}) = 1$.

\begin{lemma}
	If $Y_{\rm min}$ has $N_{\rm b}(Y_{\rm min})= c$ rectangular blocks, then we must have
	\es{muY}{
		\m(Y)\geq c \,.
	}
	In particular $\m(Y)=1$ if and only if $Y$ is related by partial transpositions to a diagram of the type $Y_{p\times q}$, for some $p$ and $q$.
	\label{lemmabd}
\end{lemma}
\begin{proof}
	Since $\mu$ is invariant under partial transposition, we can take $Y=Y_{\rm min}$.
	By assumption, the Young diagram $Y$ has a minimum of $c$ rectangular blocks in a horizontal decomposition. Correspondingly, $Y$ has $c$ outward-pointing corners along the North-East boundary.
	We obtain a finer decomposition of $Y$ into $1 + 2 + \cdots + c = \frac{c(c+1)}{2}$ smaller rectangular blocks by drawing perpendicular lines from the $c$ corners in an obvious fashion (see Figure~\ref{fig:recdec}).
	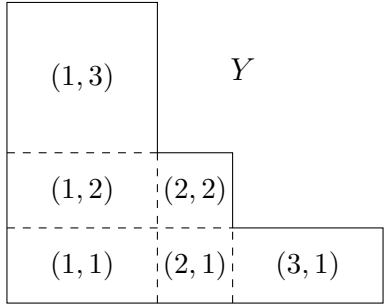
\begin{figure}[htb!]
		\centering
		\begin{tikzpicture}
		\draw (0,0) -- ++(5,0) --++(0,1) -- ++(-2,0) --++(0,1) --++(-1,0)--++(0,2)--++(-2,0) -- cycle
		node[above right=4cm]{$Y$};
		\draw [dashed] (0,1) -- ++(3,0);
		\draw [dashed] (0,2) -- ++(2,0);
		\draw [dashed] (2,0) -- ++(0,2);
		\draw [dashed] (3,0) -- ++(0,1);
		\node  at (1,.5) (a) {\small $ (1,1)$};
		\node  at (1,1.5) (a) {\small $(1,2)$};
		\node  at (2.5,1.5) (a) {\small $(2,2)$};
		\node  at (1,3) (a) {\small $(1,3)$};
		\node  at (2.5,.5) (a) {\small $ (2,1)$};
		\node  at (4,.5) (a) {\small $ (3,1)$};
		\end{tikzpicture}
		\caption{Decomposition of the Young diagram $Y$ into rectangular blocks by perpendicular lines from the $N_{\rm b}=3$ corners. The individual blocks are labelled as $(i,j)$ with $i+j\leq N_{\rm b}+1=4$ here.}
		\label{fig:recdec}
	\end{figure}
	Each of the $c$ blocks at the corners contributes $2$ or $1$ to $\mu$ depending on whether it has equal or non-equal sides. Therefore, it suffices to show that the $c(c-1)\over 2$ interior blocks each contribute non-negatively to $\mu(Y)$. 
	
	We label these rectangular blocks as $(i,j)$ for $1\leq i,j \leq c$ and $i+j \leq c+1$, and let their sizes be $p_j \times q_i$. The corner blocks are given by those with $i+j=c+1$. Each interior rectangular block labeled by $(i,j)$  contributes at least $-1$ to $\mu(Y)$, and this happens precisely when either $p_j+p_{j+1}=q_{i+1}$, or $q_i+q_{i+1}=p_{j+1}$ (see Figure~\ref{fig:reducingNb}). Let's assume this is the case for the interior rectangular block $(i,j)$. But this means $N_b$ can be reduced by $1$ after a transposition of the subdiagram involving the rectangular blocks $(k,l)$ with $k>i,l\geq j$  if $p_j+p_{j+1}=q_{i+1}$, and similarly from transposing the subdiagram involving the rectangular blocks $(k,l)$ with $k\geq i,l> j$ if  $q_i+q_{i+1}=p_{j+1}$. Since such operations are all achievable by a sequence of partial transpositions, this contradicts with the fact that $Y=Y_{\rm min}$ minimizes $N_\text{b}$. Thus each interior rectangular block can only contribute non-negatively to $\mu(Y)$.
	Therefore $\mu(Y)$ is bounded from below by the contributions of the $N_{\rm b}=c$ corner blocks and the lemma follows.
\end{proof}
\begin{figure}[htb!]
	\centering
	\begin{tikzpicture}
	\node at (0,0) (a) {};
	\draw [fill=orange,orange] ($(a)+(1,0)$) rectangle ++(3,2);
	\draw (0,0) -- ++(0,3) --++(1,0)   --++(0,-1)--++(3,0)--++(0,-2) -- cycle
	node[above right=4cm]{};
	\draw [dashed] (1,0) -- ++(0,3);
	\draw [dashed] (0,2) -- ++(2,0);
	\draw[<->] ($(a)+(1.1,0)$) -- ++(0,2) node[midway, right] {\small $p_j$};
	\draw[<->] ($(a)+(1.1,2)$) -- ++(0,1) node[midway, right] {\small $p_{j+1}$};
	\draw[<->] ($(a)+(0,-.1)$) -- ++(1,0) node[midway, below] {\small $q_{i}$};
	\draw[<->] ($(a)+(1,-.1)$) -- ++(3,0) node[midway, below] {\small $q_{i+1}$};
	\draw (0.5,0) rectangle ++(.5,.5) node[pos=.5] {\tiny $1$};
	\node at ($(a)+(5.5,0)$) (b) {};
	\draw [fill=orange,orange] ($(b)+(1,0)$) rectangle ++(2,3);
	\draw ($(b)$) rectangle ++(3,3);
	node[above right=4cm]{};
	\draw [dashed] ($(b)+(1,0)$) -- ++(0,3); 		
	\draw[->] ($(a)+(4.2,2)$) -- ($(b)+(-.4,2)$) node[midway, above=.1] {\small PT};
	\end{tikzpicture}
	~~
	\begin{tikzpicture}
	\node at (0,0) (a) {};
	\draw [fill=orange,orange] ($(a)+(0,1)$) rectangle (2,4);
	\draw (0,0) -- ++(3,0) --++(0,1)   --++(-1,0)--++(0,3)--++(-2,0) -- cycle
	node[above right=4cm]{};
	\draw [dashed] (0,1) -- ++(3,0);
	\draw [dashed] (2,0) -- ++(0,2);
	\draw (0,.5) rectangle (.5,1) node[pos=.5] {\tiny $-1$};
	\node at ($(a)+(3.5,0)$) (b) {};
	\draw [fill=orange,orange] ($(b)+(0,1)$) rectangle ++(3,2);
	\draw ($(b)$) rectangle ++(3,3);
	node[above right=4cm]{};
	\draw [dashed] ($(b)+(0,1)$) -- ++(3,0); 		
	\draw[->] ($(a)+(2.4,2)$) -- ($(b)+(-.5,2)$) node[midway, above=.1] {\small PT};
	\draw[<->] ($(a)+(1.9,0)$) -- ++(0,1) node[midway, left] {\small $p_j$};
	\draw[<->] ($(a)+(1.9,1)$) -- ++(0,3) node[midway, left] {\small $p_{j+1}$};
	\draw[<->] ($(a)+(0, 1.1)$) -- ++(2,0) node[midway, above] {\small $q_{i}$};
	\draw[<->] ($(a)+(2, 1.1)$) -- ++(1,0) node[midway, above] {\small $q_{i+1}$};
	\end{tikzpicture}
	\caption{Reduction of $N_b$ (or corners) by partial transpositions. Note that the $\rm PT$ in the right diagram involves a sequence of three partial transpositions. 
	}
	\label{fig:reducingNb}
\end{figure}
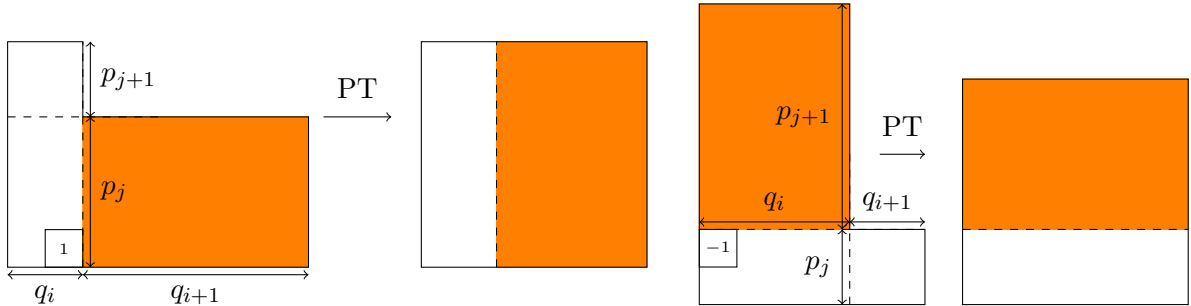

We will also need the following lemma concerning the general properties of terms that appear in the summand of \eqref{genNek}.
\begin{lemma}
	\label{lemmaGt}
	For a pair of Young diagrams $Y_1$ and $Y_2$, the following function 
	\ie
	G_{Y_1,Y_2}(a)  \equiv 
	\prod_{s\in Y_1 } {
		(a-(h_{Y_1}(s)-v_{Y_2}(s)))
	}
	\prod_{t\in Y_2 } {
		(-a-(h_{Y_2}(t)-v_{Y_1}(t))
	}
	\fe
	satisfies
	\ie
	G_{Y_1,Y_2}(a) =	G_{Y_1,Y_2^T}(a) 
	\fe
\end{lemma}

\begin{proof}
	We use the following identity from Appendix A.1 of \cite{Kanno:2013aha},
	\ie
	G_{Y_1,Y_2}(a) &=
	{\prod_{(\A,\B)\in Y_1} (a+1+M-\A -\B)
		\prod_{(\A,\B)\in Y_2} (-a+1+M-\A -\B)
		\over 
		\prod_{\A=\B=1}^{M} (a+\A-\B)
	}
	\\
	&\times 
	\prod_{\A=\B=1}^{M} (a-(h_{Y_1}((\A,\B)))-v_{Y_2}((\A,\B)))
	\label{idforG}
	\fe
	which holds when $M$ is any positive integer larger than the widths and heights of $Y_1$ and $Y_2$. Clearly, the term in the first line of \eqref{idforG} is invariant under $Y_2 \to Y_2^T$, so it suffices to show that the last factor is also invariant. This in turn follows from the simple identity (see Lemma~\ref{lemmadiaginv} and Figure~\ref{fig:slidediag})
	\ie
	\prod_{\A=1}^M (\A + \rho_\A)
	=
	\prod_{\A=1}^M (\A + \rho_\A^T)
	\label{YTid}
	\fe
	where $\rho_\A$ are column lengths for $Y_2=[\rho_1,\dots, \rho_m]$ (and zero for $\A$ larger than the width of $Y_2$).
	
\end{proof}

The strategy of the proof for the main theorem \ref{thmvan} is as follows. We first focus on the case where $\vec Y$ contains a single nonempty Young diagram $Y_{\hat i}=Y$. By studying the structure of the summand $Z_{\vec Y}$ in \eqref{genNek} in the limit $\epsilon_{1,2}\to 1$, we argue that the contribution  to \eqref{m2nek} vanishes if $\mu(Y)>2$. Furthermore, if $\mu(Y)=2$ and $Y$ is not related by partial transpositions to a rectangular Young diagram, its contribution to \eqref{m2nek} cancels that of its partner from an involution with respect to a particular subdiagram which we will define. Consequently, only the cases with $Y$ related to a rectangular Young diagram by partial transpositions (either with $\mu(Y) = 2$ or $\mu(Y) = 1$, since in the latter case all Young diagrams are related by partial transpositions to a rectangular one) matter for \eqref{m2nek}. Next, we extend the analysis to general $\vec Y$. Once again, the structure of the summand in \eqref{genNek} implies that the only other relevant case corresponds to $\vec Y$ having two non-empty Young diagrams.  Up to a reordering of the $Y_i$, let us consider $\vec{Y} = (Y_1, Y_2, \varnothing,\dots,\varnothing)$. Moreover, one needs $\mu(Y_1)=\mu(Y_2)=1$ for a nonvanishing contribution to \eqref{m2nek}. We then show that the contribution from $\vec Y=(Y_1,Y_2,\varnothing,\dots,\varnothing)$ to \eqref{m2nek} (though non-vanishing) cancels with that from $\vec Y=(Y_1,Y_2^T,\varnothing,\dots,\varnothing)$. This completes the proof of Theorem~\ref{thmvan}. Below we explain each step in the above order in more detail.

We first take $\vec Y$ to be such that the only nonempty Young diagram is $Y_{\hat i}=Y$.  The quantity $Z_{\vec{Y}}$ defined in \eqref{genNek} takes the form
\ie
Z_{\vec Y}
=
F_1(Y) F_2(Y)
\label{ZYdec1}
\fe
where
\ie 
F_1(Y)  &=
\prod_{s\in Y} {
	(E(0,Y,Y,s)-i m -\epsilon_+/2) (-E(0,Y,Y,s)-i m +\epsilon_+/2)
	\over  
	E(0,Y,Y,s)  (\epsilon_+-E(0,Y,Y,s))
} \,,
\\
F_2(Y)  &= \prod_{ j=1,j\neq \hat i}^N	 
\prod_{s\in Y } {
	(E(a_{\hat ij},Y,\varnothing,s)-i m -\epsilon_+/2) 
	\over  
	E(a_{\hat ij},Y,\varnothing,s) 
}
{(E(a_{\hat ij},Y,\varnothing,s)+i m -\epsilon_+/2) 
	\over  
	E(a_{\hat ij},Y,\varnothing,s) -\ep_+
}\,.
\label{ZYdec2}
\fe
Our goal here is to understand properties of $Z_{\vec Y}$ in the limit $\epsilon_-\to 0$ (and $\epsilon_+\to 2$) in relation to the shape of the Young diagram $Y$.

The quantity $F_2(Y)$ is manifestly finite and nonzero (for generic $a_i$) for all $m, \ep_\pm$, while $F_1(Y)$ has a subtler behavior in the limit $\ep_-\to 0$.  To understand how $F_1(Y)$ behaves in this limit, we write it as
\ie
F_1(Y)=
F_1^0 (Y) F_1^+(Y) F_1^-(Y)F_1^r(Y) \,,
\fe
where we have decomposed the product over $s\in Y$ as $Y=Y_0\sqcup Y_+\sqcup Y_-\sqcup Y_{r}$ defined by 
\es{Y0pmrDefs}{
	Y_0 &= \{s\in Y |h(s)-v(s)=0 \} \,,
	\\
	Y_\pm &= \{s\in Y |h(s)-v(s)=\pm 1 \} \,,
	\\
	Y_r &= Y \backslash (Y_0\sqcup Y_+ \sqcup Y_-) \,,
}
and $F_1^0,F_1^\pm, F_1^r$ correspond to products over these disjoint subsets, respectively.  Note that according to our previous definition,  $n_0=|Y_0|$ and $n_{\pm 1}=|Y_\pm|$. More explicitly, we have
\ie
F_1^0 &= \prod_{s\in Y_0} {(h \ep_-+{\ep_-\over 2} + m)(h \ep_-+{\ep_-\over 2} - m)  \over (h\ep_- -\ep_2) (h \ep_-+\ep_1)} \,,
\\
F_1^+ &= \prod_{s\in Y_+} {(h \ep_-+{\ep_+\over 2} + m)(h \ep_-+{\ep_+\over 2} - m)  \over h\ep_- (h \ep_-+\ep_+)} \,,
\\
F_1^- &= \prod_{s\in Y_-} {(h \ep_-+{\ep_+\over 2} + m -2\ep_2)(h \ep_-+{\ep_+\over 2} - m -2\ep_2)   \over (h\ep_- -2\ep_2) (h+1)\ep_- } \,,
\fe
where, for notational simplicity, we suppressed the $s$ dependence in $h(s)$. For the quantity \eqref{m2nek} of interest, we need to take the limit $\epsilon_-\to 0$. In this limit, $F_1^0$ could potentially vanish, $F_1^\pm$ could potentially blow up, and $F^r_1$ is finite and nonzero. 

Expanding in small $\ep_-$ (and taking $\epsilon_+=2$), we have 
\ie
F_1^0 &=  \prod_{ s\in Y_0} \left(-m^2-(h+1/2)^2\ep_-^2\right) (1+\cO(\ep_-^2)) \,,
\\
F_1^+ &= \prod_{ s\in Y_+} {1\over 2 h \ep_-} \left(1+{3\over 2} h \ep_-+\cO(\ep_-^2)\right ) \,,
\\
F_1^- &= \prod_{ s\in Y_-} {1\over -2 (h+1) \ep_-} \left (1-{3\over 2} (h+1) \ep_-+\cO(\ep_-^2)\right ) \,.
\fe
From lemma~\ref{lemmabd}, $\mu(Y)=2n_0 -n_1-n_{-1}\geq 1$, which implies that there can be at most a simple pole in $\ep_-$ from $F_1$ at order $m^2$. For the purpose of extracting the order $m^2$ term from $Z_Y$,  we can take the truncation (which we denote by $\simeq$)
\ie
F_1^0 \simeq & \prod_{ s\in Y_0} \left(-m^2-(h+1/2)^2\ep_-^2\right) ,
\\
F_1^+\simeq& \prod_{ s\in Y_+} {1\over 2 h \ep_-} \left(1+{3\over 2} h \ep_-\right ),
\\
F_1^- \simeq & \prod_{ s\in Y_-} {1\over -2 (h+1) \ep_-} \left (1-{3\over 2} (h+1) \ep_-\right ),
\\
F_1^r \simeq &\left. F_1^r\right|_{m=0}.
\label{F1form}
\fe
We recollect the relevant contributions to $F_1$ as
\ie
F_1(Y)\simeq m^2\ep_-^{\mu(Y)-2}(F_1^{(0)}+\ep_- F_1^{(1)})F_1^r
\label{F1epm}
\fe
where 
\ie
F_1^{(0)}(Y) = {1 \over 
	\prod_{s\in Y_+} (2h) \prod_{s\in Y_-} (-2v)
} \prod_{s\in Y_0} \left(h+{1\over 2}\right)^2 \sum_{s\in Y_0} {1\over \left(h+{1\over 2}\right)^2}
\label{F10form}
\fe
and
\ie
F_1^{(1)}(Y) = 
{3\over 2} F_1^{(0)}(Y) (\sum_{s\in Y_+}  h-\sum_{s\in Y_-}  v) \,.
\fe
In addition,
\ie
F_1^r=F_1^{r(0)}+\ep_- F_1^{r(1)} \,,
\fe
where
\ie
F_1^{r(0)}(Y) =& 
\prod_{s\in Y } {
	(h-v) ^2
	\over  
	(h-v)^2-1
} \,.
\fe
Similarly, it suffices to set $m=0$ for $F_2$, 
\ie
F_2(Y)\simeq F^{(0)}_2(Y)+ \ep_- F^{(1)}_2(Y)   \,, 
\fe
where we only kept the terms up to first order in $\ep_-$.
Here 
\ie
F^{(0)}_2(Y)
=&\prod_{ j=1,j\neq \hat  i}^N	 
\prod_{s\in Y } {
	(i a_{j\hat i} -(h_{\varnothing}(s)-v_Y(s)) ^2
	\over  
	(i a_{j\hat i} -(h_{\varnothing}(s)-v_Y(s))^2-1
}
\fe
and 
\ie
&F^{(1)}_2(Y)
\\
=&
F^{(0)}_2(Y)
\sum_{ j=1,j\neq \hat i}^N	 
\sum_{s\in Y }
{
	h_{\varnothing}(s)+v_Y(s)+1 
	\over  
	(i a_{j\hat i} - (h_{\varnothing}(s)-v_Y(s) )(i a_{j\hat i} - (h_{\varnothing}(s)-v_Y(s)+1 ))(i a_{j\hat i} - (h_{\varnothing}(s)-v_Y(s)-1 ))
} \,.
\fe

We have thus spelled out which parts of  \eqref{ZYdec1} and \eqref{ZYdec2} matter for evaluating the  limit $\epsilon_{1,2}\to 1$ needed in \eqref{m2nek}, and how they depend on the combinatoric properties of the Young diagram $Y$. We are ready to see which Young diagrams can potentially contribute to \eqref{m2nek}.
From Lemma~\ref{lemmabd}, any Young diagram $Y$ that is not related by partial transposition to a rectangular one has $\mu(Y)\geq 2$. From \eqref{F1epm}, the contribution to \eqref{m2nek} from this particular Young diagram vector $\vec Y$ (with the single non-empty Young diagram $Y$) goes as,
\ie
I_{k,\vec Y} \sim \lim_{\ep_-\to 0} \ep_-^{\mu(Y)-2}
\fe
Thus if  $\mu(Y)>2$, it vanishes identically. We then move onto cases of  $Y$ with $\mu(Y)=2$, and its contribution to \eqref{m2nek} is then deduced from
\ie
F_1(Y)\simeq m^2 F_1^{(0)} (Y)  \,, \qquad 
F_2 (Y)\simeq F_2^{(0)}(Y) 
\fe
to be
\ie
I_{k,\vec Y} =2F_1^{(0)} (Y)F_2^{(0)} (Y) \,.
\fe
It is easy to see that the corresponding $Y_{\rm min}$ takes the form as in Figure~\ref{fig:Ymin} from stacking $3$ rectangular boxes of sizes $p_1\times q_1$, $p_2\times q_1$, and $p_1\times q_2$ respectively where the sides satisfy (see proof of Lemma~\ref{lemmabd})
\ie
q_2\geq p_1+p_2+1\,,\qquad q_1>p_2\,.
\fe
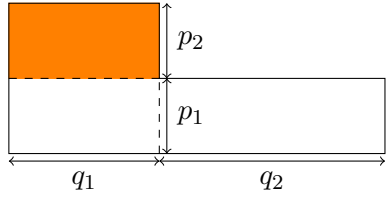
\begin{figure}[htb!]
	\centering
	\begin{tikzpicture}
	\node at (0,0) (a) {};
	\draw [fill=orange,orange] ($(a)+(0,1)$) rectangle ++(2,1);
	\draw (0,0) -- ++(0,2) --++(2,0)   --++(0,-1)--++(3,0)--++(0,-1) -- cycle
	node[above right=4cm]{};
	\draw [dashed] (2,0) -- ++(0,1);
	\draw [dashed] (0,1) -- ++(2,0);
	\draw[<->] ($(a)+(2.1,0)$) -- ++(0,1) node[midway, right] {\small $p_1$};
	\draw[<->] ($(a)+(2.1,1)$) -- ++(0,1) node[midway, right] {\small $p_{2}$};
	\draw[<->] ($(a)+(0,-.1)$) -- ++(2,0) node[midway, below] {\small $q_{1}$};
	\draw[<->] ($(a)+(2,-.1)$) -- ++(3,0) node[midway, below] {\small $q_{2}$};
	\end{tikzpicture}
	\caption{ 
		$Y_{\rm min}$ for a Young diagram $Y$ with $\m(Y)=2$ and not related to a rectangular diagram by partial transpositions.
	}
	\label{fig:Ymin}
\end{figure}
Under partial transpositions, we will show that the $p_2\times q_1$ block of $Y_{\rm min}$ gets mapped to a Young subdiagram of the resulting Young diagram $Y$. We define $W^Y$ to be this particular subdiagram in $Y$ (see Figure~\ref{fig:ptW} for an example).   Then we define the Young diagram  $Y'$ (correspondingly the Young diagram vector $\vec Y'$ with a single non-empty Young diagram $Y'_{\hat i}=Y'$) to be obtained from $ Y$ by a sequence of partial transpositions  that send $W^Y$ to its transpose while keeping the rest of the diagram fixed.(It is always possible to find such a sequence of partial transpositions.)\footnote{The existence of such a sequence of partial transpositions is guaranteed for $Y$ of the form in Figure~\ref{fig:Yform} (with the orange block $W^Y$ possibly replaced by its partial transpositions). In the later part of the section, we will prove that for $Y$ with $\mu(Y)=2$ and not related to rectangular diagrams by partial transpositions, it takes the form as in Figure~\ref{fig:Yform}. Here we assume this is the case. Note that $\mu(Y)=2$ demands the bottom-left $Y_{t\times r}$ sub-diagram (gray block) to contribute $-1$ to $\mu(Y)$  since each of the three exterior colored blocks along the North-East boarder (orange or gray) contribute at least 1 to $\mu(Y)$. This means we have either  $\lambda_1^T-r=\lambda_r$ or $\lambda_t^T=\lambda_1-t$(see proof of Lemma~\ref{lemmabd}). In the former case, $\iota_Y$ is achieved by ${\rm PT}_{r+1}$ followed by ${\rm PT}_{r+t+1}$; in the latter case, $\iota_Y={\rm PT}_1 \cdot  {\rm PT}_{r+t+1} \cdot  {\rm PT}_{t+1}\cdot  {\rm PT}_1 $. 
}
We denote this involution by  $\iota_Y$.

\begin{figure}[htb!]
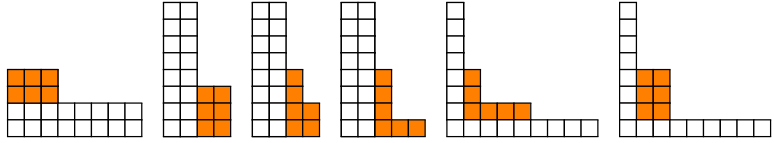

	\centering
 
\ytableausetup
{boxsize=.5em}	\ydiagram[*(orange)  ]
{3,3}
*[*(white)]{3,3,8,8}
~
\ydiagram[*(white)  ]
{2,2,2,2,2,2,2,2}
*[*(orange)]{2,2,2,2,2,4,4,4}
~
\ydiagram[*(white)  ]
{2,2,2,2,2,2,2,2}
*[*(orange)]{2,2,2,2,3,3,4,4}
~															
\ydiagram[*(white)  ]
{2,2,2,2,2,2,2,2}
*[*(orange)]{2,2,2,2,3,3,3,5}
~	\ydiagram[*(white)  ]
{1,1,1,1,1,1,1,9}
*[*(orange)]{1,1,1,1,2,2,5}
~	\ydiagram[*(white)  ]
{1,1,1,1,1,1,1,9}
*[*(orange)]{1,1,1,1,3,3,3}																	

	\caption{
An example with $Y_{\rm min}=[4,4,4,2,2,2,2,2]$ (first diagram above). The other $Y$'s are related by partial transpositions to $Y_{\rm min}$ and the subdiagram $W^Y$ is marked in orange. }
	\label{fig:ptW}
\end{figure}

We show below that
\ie
I_{k,\vec Y}+I_{k,\vec Y'}=0 \,.
\fe
This is a consequence of the properties of $F_1^{(0)} (Y)$ and $F_2^{(0)} (Y)$ under  partial transpositions: the latter is even whereas the former is odd for the particular type of $Y$ and involution $\iota_Y$ we consider here.

We first show that $F_2^{(0)}(Y)$ is invariant under partial transposition with respect to an arbitrary subdiagram $P$.  We have the decomposition
\ie
F^{(0)}_2(Y)=\prod_{ j=1,j\neq \hat i}^N	 
\prod_{s\in Y \backslash P } {
	(i a_{j\hat i} -(h_{\varnothing}(s)-v_Y(s)) ^2
	\over  
	(i a_{j\hat i} -(h_{\varnothing}(s)-v_Y(s))^2-1
}
\prod_{s\in  P } {
	(i a_{j\hat i} -(h_{\varnothing}(s)-v_Y(s)) ^2
	\over  
	(i a_{j\hat i} -(h_{\varnothing}(s)-v_Y(s))^2-1
} \,.
\label{decomp}
\fe
Since for general $s=(\A,\B)\in Y$ and $Y=[\lambda_1,\dots, \lambda_l]$
\ie
h_{\varnothing}(s)-v_Y(s)=-\A-(\lambda_\A-\B)  \,, 
\fe
where $\lambda_\B^T$ does not appear,  the first factor in \eqref{decomp} from $Y \backslash P$ does not change under $P\to P^T$. It is also easy to see that the second factor is also invariant. For $(\A,\B)\in P$ and $(\B,\A)\in P^T $ with $P=[\rho_1,\dots, \rho_m]$,
\ie
h_{\varnothing}(\A,\B)-v_Y(\A,\B) &= -\A-(l-m)-\rho_{\A}+\B \,, 
\\
h_{\varnothing}(\B,\A)-v_{{\rm PT}(Y)}(\B,\A)  &= -\B-(l-m)-\rho^T_{\B}+\A 
\fe
are mapped to each another through
\ie
\A \to 1+\rho_\B^T-\A \,, \qquad \B \to 1+\rho_\A-\B \,.
\fe

Thus we've shown
\ie
F^{(0)}_2(Y)=F^{(0)}_2({\rm PT}(Y)) \,.
\label{F20pf}
\fe

Let's now show that 
that $F_1^{(0)}(Y)$ is odd under the involution $\iota_Y$ if $\mu(Y)=2$. We rewrite $F_1^{(0)}(Y)$ as
\ie
F_1^{(0)}(Y)=&f_1(Y) f_2(Y) f_3(Y)
\\
f_1=&{1 \over 
	\prod_{s\in Y_+} (2h) \prod_{s\in Y_-} (-2v)
} 
\\
f_2=&\prod_{s\in Y_0} \left(h+{1\over 2}\right)^2 \sum_{s\in Y_0} {1\over \left(h+{1\over 2}\right)^2}
\\
f_3=& 
\prod_{s\in Y_r } {
	(h-v) ^2
	\over  
	(h-v)^2-1
} \,.
\fe
Once again, we separate the products above as
\ie
f_i (Y)=\prod_{s\in Y \backslash W^Y} \dots \prod_{s\in W^Y}  \dots \equiv  f_i^1(Y)  f_i^2(Y)
\fe
where $f_i^1$ and $f_i^2$ only receives contributions from entries in $Y\backslash W^Y$ and $W^Y$ respectively. Since $\iota_Y$ simply transposes $W^Y$, which exchanges $h$ and $v$, we have 
\ie
f_1^2(Y) =-f_1^2(\iota_Y(Y)) \,, \qquad f_2^2(Y)=f_2^2(\iota_Y(Y)) \,,
\qquad f_3^2(Y)=f_3^2(\iota_Y(Y)) \,.
\fe
Now if $Y$ takes the form as in Figure~\ref{fig:Yform} with $W^Y$ a rectangular block (or type either $p_2\times q_1$ or its transpose as in Figure~\ref{fig:Ymin}), then clearly $f^1_i$ are invariant under the involution $\iota_Y$ since the values of $h$ (or $v$) do not change for boxes in the gray parts of the Young diagram as in Figure~\ref{fig:Yform}. Furthermore, if $W^Y$ is replaced  its partial transpositions (leaving the rest of $Y$ fixed), the contributions to $f^1_i$ still come from the same gray blocks as in Figure~\ref{fig:Yform}, and consequently invariance of $f^1_i$ under $\iota_Y$ continues to hold thanks to corollary~\ref{coro:Ysize} (see also proof of Lemma~\ref{lemmainv}). Since $Y$ is related to $Y_{\rm min}$ in Figure~\ref{fig:Ymin} by a sequence of partial transpositions, it is sufficient to show that any partial transposition preserves the form of $Y$ up to partial transpositions restricted to the subdiagram $W^Y$ in Figure~\ref{fig:Yform} (in particular $W^Y$  remains a Young subdiagram under any partial transposition of $Y$).

We proceed by induction. First note $Y_{\rm min}$ is of the form in Figure~\ref{fig:Yform} with $r=0,t=p_1$ and $p=p_2,q=q_1$. Now suppose we start with a general $Y$ as in Figure~\ref{fig:Yform}, we show that any partial transposition PT$_\A$ preserves its form, namely only the gray blocks contain boxes with $h-v=0$ or $\pm 1$ among all boxes in $Y\backslash W^Y$. Clearly, the only nontrivial partial transpositions are PT$_\A$ with $\A=1$ which corresponds to the usual transposition, or $1<\A\leq r$ and $\lambda_1^T-\A+1\leq \lambda_{\A-1}$,  or  $\lambda_1^T\geq \A\geq r+q+1$ and $\lambda^T_1-\A+1\leq t$. The first case clearly preserves the form of $Y$ as in Figure~\ref{fig:Yform}. In the second case, the induction hypothesis implies the values $h-v$ for the boxes just below the top left gray block in Figure~\ref{fig:Yform} satisfy
\ie
r-\A +(q+p+t)-\lambda_\A \leq -2
\label{ind1}
\fe
for $1\leq \A \leq r$, and this is sufficient to ensure that the left white block in Figure~\ref{fig:Yform} does not contain entries with $h-v$ equal to $0$ or $\pm 1$. Suppose we perform a partial transposition PT$_x(Y)$ for some $1< x\leq r$, then the resulting Young diagram is again of the form in Figure~\ref{fig:Yform} with the change in the parameters,
\ie
(p,q,t,r)\to (q,p,r-x+1,t+x-1)
\fe
and to ensure that the (new) left white block as in Figure~\ref{fig:Yform} does not contain entries with $h-v$ equal to $0$ or $\pm 1$, we need
\ie
t+x-1-\A +(q+p+t-x+1)-\lambda_\A \leq -2 \,,
\fe
which follows from \eqref{ind1}. A similar argument applies to the third case, focusing on the right bottom white block in Figure~\ref{fig:Yform}.

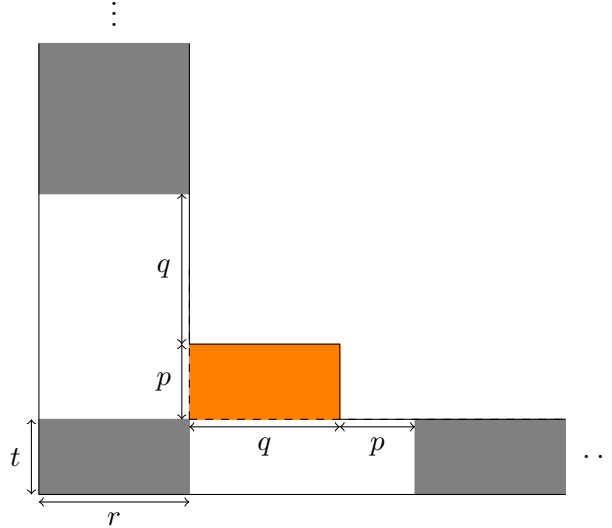
\begin{figure}[htb!]
	\centering
	\begin{tikzpicture}
	\node at (0,0) (a) {};
	\draw [fill,orange] ($(a)+(2,1)$) rectangle ++(2,1);
	\draw [fill,gray] ($(a)+(0,4)$) rectangle ++(2,2);
	\draw [fill,gray] ($(a)+(5,0)$) rectangle ++(2,1);
	\draw [fill,gray] ($(a)$) rectangle ++(2,1);
	\draw (0,0) -- ++(0,6);
	\draw (0,0) -- ++(7,0);
	node[above right=4cm]{};
	\draw [dashed] (2,1) -- ++(0,2);
	\draw [dashed] (2,1) -- ++(5,0);
	\draw  (2,6) -- ++(0,-4)-- ++(2,0)-- ++(0,-1)-- ++(3,0);
	\draw[<->] ($(a)+(2-.1,1)$) -- ++(0,1) node[midway, left] {\small $p$};
	\draw[<->] ($(a)+(2-.1,2)$) -- ++(0,2) node[midway, left] {\small $q$};
	\draw[<->] ($(a)+(4,1-.1)$) -- ++(1,0) node[midway, below] {{\small $p$}};
	\draw[<->] ($(a)+(0,-.1)$) -- ++(2,0) node[midway, below] {\small $r$};
	\draw[<->] ($(a)+(0-.1,0)$) -- ++(0,1) node[midway, left] {\small $t$};
	\draw[<->] ($(a)+(2,1-.1)$) -- ++(2,0) node[midway, below] {\small $q$};
	\node at (7.5,0.5)   {$\dots$};
	\node at (1,6.5)  {$\vdots$};
	\end{tikzpicture}
	\caption{ 
		Young diagram $Y$ related to $Y_{\rm min}$ in Figure~\ref{fig:Ymin} by partial transpositions. The orange block denotes $W^Y$. The gray blocks contain all the boxes with $h-v=0$ or $\pm 1$ in $Y\backslash W^Y$. Note that $r,t$ are non-negative integers. If $t=0$, the only gray block is the one on the top left; if $r=0$, the only gray block is the one on the bottom right. In particular $Y_{\rm min}$ in Figure~\ref{fig:Ymin} corresponds to the special case $r=0,t=p_1$ and $p=p_2,q=q_1$. 
	}
	\label{fig:Yform}
\end{figure}

We have thus proven that
\ie
F_1^{(0)}(Y)=-F_1^{(0)}(\iota_Y(Y)) 
\label{F10pf}
\fe
for the special type of Young diagram $Y$ and involution $\iota_Y$ considered here.

Putting together the facts that $F_1^{(0)}(Y)$ and $F_2^{(0)}(Y)$ are respectively odd and even under the involution $\iota_Y$, it follows that the contributions to \eqref{m2nek} from $Y$ and $\iota_Y(Y)$ cancel pairwise in the $\epsilon_{-}\to 0$ limit as a consequence of \eqref{F10pf} and \eqref{F20pf}.  We have now shown that in the case where $\vec Y$ contains a single non-empty Young diagram, if this Young diagram is not related to a rectangular diagram by partial transpositions, it does not contribute to \eqref{m2nek}.

We now explain  why only $\vec Y$ with a single non-empty Young diagram contributes to \eqref{m2nek}. The form of $F_1$ and the behavior \eqref{F1epm} implies the only other relevant case involves two  non-empty Young-diagram in $\vec Y$ which we can take to be  
\ie
\vec Y=\{Y_1,Y_2,\varnothing,\dots,\varnothing\} \,,
\fe
up to an $S_N$ permutation, and we must have $\mu(Y_1)=\mu(Y_2)=1$ (i.e.~related to non-square rectangular Young diagram by partial transpositions). In this case, we have pair-wise cancellations between such $\vec Y$ and the partner $\vec Y'$ with $Y_2$ (or $Y_1$) replaced by its transpose  
\ie
\vec Y'=\{Y_1,Y_2^T,\varnothing,\dots,\varnothing\} \,.
\fe
To see this, recall that the summand $Z_{\vec Y}$ of \eqref{genNek} in this case decomposes as the following product
\ie
Z_{\vec Y} = Z_{Y_1} Z_{Y_2} Z_{Y_1,Y_2} \,,
\fe
where each factor only depends on the Young diagram(s) in the subscript (e.g. $Z_{Y_1}$ only involves the product in \eqref{genNek} with $i=1$ and $j\neq 2$ or $j=1$ and $i\neq 2$). From the argument we presented for the case $\vec Y=\{Y,\varnothing,\dots,\varnothing\}$, it's easy to see that $Z_{Y_1}$ is invariant under $Y_2\to Y_2^T$ while $Z_{Y_2}$ is odd since $\mu(Y_2)=1$ (see equation \eqref{F1form} and \eqref{F10form}). Thus it suffices to prove that in the limit \eqref{m2nek}
\ie
Z_{Y_1,Y_2}
\simeq 
&\prod_{s\in Y_1 } {
	(i a_{21} -(h_{Y_1}(s)-v_{Y_2}(s)) ^2 
	\over  
	(i a_{21} -(h_{Y_1}(s)-v_{Y_2}(s))^2-1
}
\prod_{t\in Y_2 } {
	(i a_{12} -(h_{Y_2}(t)-v_{Y_1}(t))^2
	\over  
	(i a_{12} -(h_{Y_2}(t)-v_{Y_1}(t))^2-1
}
\fe
satisfies
\ie
Z_{Y_1,Y_2} =Z_{Y_1,Y_2^T} \,,
\label{Y1Y2t}
\fe 
which indeed follows from  Lemma~\ref{lemmaGt}.

We have thus finished the proof of theorem~\ref{thmvan}: the limit \eqref{m2nek} of double mass derivatives of the Nekrasov partition function \eqref{genNek} is dominated by $\vec Y$ with a single non-empty Young diagram of rectangular shape (and its partial transpositions).

\section{Recursion relations} \label{sec:recursion}

In this appendix we will give a proof of (\ref{eq:pq-tab}) for the special cases in which the nontrivial Young diagram in the vector  $\vec Y$ in \eqref{rectangular} is of the form  $ {Y}_{1\times k}$ or  $ {Y}_{k \times 1}$. Recall a Young diagram vector $\vec{Y}$'s contribution to the $k$-instanton partition function is given by
  \es{Ik-gen-app}{
Z_{k, \vec{Y}}(m,a_{ij}, \epsilon_i) &= {1 \over k!} \left( \epsilon_+ (m^2 + \epsilon^2_-/4 ) \over 
\epsilon_1 \epsilon_2 (m^2 + \epsilon^2_+/4) \right)^k
\oint \prod_{I=1}^k {d\phi_I\over 2\pi}
\prod_{j=1}^N {(\phi_I -a_j)^2 - m^2  \over (\phi_I  -a_j)^2 +  \epsilon_+^2/4}
\\
&\times
\prod_{I<J}^k { \phi_{IJ}^2[\phi_{IJ}^2 + \epsilon_+^2 ] [\phi_{IJ}^2+(i m-\epsilon_-/2)^2][\phi_{IJ}^2+(i m+\epsilon_-/2)^2]\over [\phi_{IJ}^2 + \epsilon^2_1][\phi_{IJ}^2 + \epsilon^2_2]
[\phi_{IJ}^2+(i m+\epsilon_+/2)^2][\phi_{IJ}^2+(i m-\epsilon_+/2)^2]} \, .
}
For the round $S^4$, by setting $\epsilon_1=\epsilon_2=1$ we have
\ie \label{kInstFull-apn}
Z_{k, \vec Y}(m,a_{ij})=&\, {1\over k!} \left( 2m^2\over m^2+1\right)^k
\oint \prod_{I=1}^k {d\phi_I\over 2\pi}
\prod_{i=1}^N {(\phi_I-a_i)^2-m^2\over (\phi_I-a_i)^2+1}
\\
&\times
\prod_{I<J}^k { \phi_{IJ}^2(\phi_{IJ}^2+4)(\phi_{IJ}^2-m^2)^2\over (\phi_{IJ}^2+1)^2 [(\phi_{IJ}-m)^2+1][(\phi_{IJ}+m)^2+1 ]} \,. 
\fe 
In this appendix, we will first keep $\epsilon_1, \epsilon_2$ arbitrary, and set them equal to $1$ at the end. That is because with general $\epsilon_1, \epsilon_2$, the instanton partition function only has simple poles and obeys a simple recursion relation \cite{Kanno:2013aha, Nakamura:2014nha}.  Indeed, it is straightforward to see that $Z_{k, \vec{Y}}(m,a_{ij}, \epsilon_i)$ satisfies the following recursion relation
 \begin{align} \label{Ik1-gen}
Z_{k+1, \vec{Y}_+}&(m, a_{ij}, \epsilon_i)=  Z_{k, \vec{Y}}(m,a_{ij}, \epsilon_i) {1\over k+1} { \epsilon_+ (m^2+\epsilon^2_-/4) \over 
\epsilon_1 \epsilon_2 (m^2 + \epsilon^2_+/4 )  }
 \oint   {d\phi \over 2\pi}
\prod_{j=1}^N {(\phi -a_j)^2 - m^2  \over (\phi  -a_j)^2 +  \epsilon_+^2/4}
\\
\!\!\!  \times
\prod_{J=1}^k  & { (\phi -\hat{\phi}_{J})^2  [(\phi -\hat{\phi}_{J})^2 + \epsilon^2_+ ] [(\phi -\hat{\phi}_{J})^2 + (im-\epsilon_-/2)^2][(\phi-\hat{\phi}_{J})^2 + (im+\epsilon_-/2)^2]\over [(\phi-\hat{\phi}_{J})^2 + \epsilon^2_1][(\phi-\hat{\phi}_{J})^2 + \epsilon^2_2]
[(\phi-\hat{\phi}_{J})^2+(im +\epsilon_+/2)^2][(\phi-\hat{\phi}_{J})^2+(im-\epsilon_+/2)^2]} \, , \nonumber
\end{align}
where $\vec{Y}_+$ (with $k+1$ boxes) is a Young diagram vector that is obtained from $\vec{Y}$ by adding one more box. The contour for the integration over $\phi$ is determined by the position of the box that we add to $\vec{Y}$ during the recursion procedure, and $\hat \phi_J$ are the poles for evaluating $Z_{k, \vec{Y}}(m,a_{ij}, \epsilon_i)$, which are determined by $\vec{Y}$ according to (\ref{YTpoles}). 

We will here consider the contributions from Young diagram vector  $\vec Y$ as in \eqref{rectangular} where the nontrivial Young diagram is a single row Young diagram ${Y}_{1\times k}$ or a single column Young diagram ${Y}_{k \times 1}$. These Young diagrams can be constructed recursively by adding one box at a time, and one can solve the recursion relation straightforwardly for the instanton partition function at order $m^2$. 

We define
\ie
I_{1 \times k}(\epsilon_1,\epsilon_2) = \partial^2_m Z_{k, (\dots,{Y}_{1\times k},\dots)}(m,a_{ij}, \epsilon_i) \big{|}_{m=0} \, .
\fe
The recursion relation for $I_{1 \times k}(\epsilon_1,\epsilon_2)$ is solved by beginning with the initial data, which we take to be the contribution from the $\vec Y$ with nontrivial Young diagram $Y_{1\times 2}$.  For simplicity, we will begin by setting $N=1$ since the argument for general $N$ is more complicated, in which case the initial data is 
\ie
I_{1 \times 2}(\epsilon_1,\epsilon_2) =&\, \oint {dz \over 2\pi}  \left[ \frac{3  (z- a_1) (z-a_1 +i  
   \epsilon_1)}{4 \epsilon_1^2  (z-a_1-i \epsilon_1) (z -a_1 +2 i \epsilon_1)} {1\over  \epsilon_1-\epsilon_2} \right. \cr
& + \left. \frac{3 \epsilon_1^4 +  (z - a_1) (z -a_1 +i \epsilon_1) \left(22 \epsilon_1^2 +8
   (z- a_1) (z - a_1 + i \epsilon_1)\right) }{4 \epsilon_1^3 (z - a_1 -i\epsilon_1)^2 (z -a_1+2 i \epsilon_1)^2} \right]  \, ,
\fe
where the integration is around the poles at $z =a_1 +i \epsilon_+/2$, but the integrand is our focus here. 
From this we find that, in general, the $\vec Y$ with nontrivial Young diagram ${Y}_{1\times k}$ gives the following contribution
\begin{align}
&\, I_{1 \times k}(\epsilon_1,\epsilon_2) =\oint {dz  \over 2\pi}   \left[  \frac{ \left( 2k^2 -2\right) (z- a_1) ( (z -a_1)+i (k-1)
   \epsilon_1)}{ \epsilon_1^2 \, k^2  k! ( (z -a_1)-i \epsilon_1) ( ( z-a_1) +i k \epsilon_1)} {1\over  \epsilon_1-\epsilon_2}  \right.  \\
\!\!\! +& \left.  \frac{ (k-1)^2 (k+1) \epsilon_1^4 +  (z- a_1) (z-a_1 +i (k-1) \epsilon_1) \left( (6k^2-2) \epsilon_1^2 +4k
   (z - a_1) (z - a_1 + i (k-1) \epsilon_1)\right) }{k\, k! \, \epsilon_1^3 (z - a_1 -i\epsilon_1)^2 (z-a_1+ i k  \epsilon_1)^2} \right] \, . \nonumber
\end{align}
Here, we only keep the terms of order ${O}((\epsilon_1- \epsilon_2)^0)$. 
The contribution from the $\vec Y$ with nontrivial Young diagram $ {Y}_{k \times 1}$ is obtained by the transposition, which simply exchanges $\epsilon_1$ with $\epsilon_2$, 
\ie
I_{k \times 1}(\epsilon_1,\epsilon_2) = I_{1\times k} (\epsilon_2,\epsilon_1) \, .
\fe
The singular terms cancel in the sum of these two contributions, and the final result is given by 
\ie \label{N=1}
I_{1\times k} & \, = \left[ I_{1\times k}(\epsilon_1,\epsilon_2) + I_{k \times 1}(\epsilon_1,\epsilon_2) \right] \big{|}_{\epsilon_{1}=\epsilon_{2}=1}  \cr
& \,  = \oint {dz \over 2\pi}
\prod_{k_a=0}^{k-1}
 {(z-a_1 + k_a i )^2\over (z-a_1 + k_a i)^2+1 } 
 \times
{1 \over k!} \left[ 
{4\over 1+\delta_{1k}} \left(1+{1\over k^2} \right) \right. \cr
& ~~~ + \left. 
{2 i (k+1)(k-1)^2 \over k(z-a_1+ki) (z-a_1+(k-1)i) (z-a_1)}
\right] \, , 
\fe
where we have set $\epsilon_1= \epsilon_2=1$ at the end of computation. Having determined the $N=1$ case, the generalization of the above integrand to arbitrary $N$ is straightforward. This can be seen by using \eqref{kInstFull-apn} as follows. 

If $I_{1\times k}$ were computed from \eqref{kInstFull-apn}, we would take consecutive residues surrounding the poles at $\phi^2_{IJ}+1=0$. Since they are higher order poles, one needs to expand the integrand when taking the residues. The constant term ${4\over 1+\delta_{1k}} \left(1+{1\over k^2} \right)$ of $I_{1\times k}$ clearly comes from the second line in \eqref{kInstFull-apn}, which is independent of $N$, whereas the term containing $z-a_1$ requires expanding the first line in \eqref{kInstFull-apn}, and its generalization to general $N$ is obvious. Therefore,  for general $N$ we have
\begin{align} \label{N=1Again}
I_{1\times k} & \,  = \oint {dz \over 2\pi}
\prod_{k_a=0}^{k-1}
\prod_{j=1}^N {(z-a_j + k_a i )^2\over (z-a_j + k_a i)^2+1 } 
 \times
{1 \over k!} \left[ 
{4\over 1+\delta_{1k}} \left(1+{1\over k^2} \right) \right. \cr
& ~~~ + \left. \sum_{j=1}^N
{2 i (k+1)(k-1)^2 \over k(z-a_j+ki) (z-a_j+(k-1)i) (z-a_j)}
\right] \, , 
\end{align}
and the expression (after summing over $k!$ identical contributions) agrees with the general formula given in (\ref{eq:pq-tab}).  

\section{Instantons at higher order in $1/N$}
\label{highInst}

In this appendix we compute the instanton contributions to $\partial_m^2 \log Z \big\vert_{m=0}$ to $O(N^{-\frac92})$ in a small $g_\text{YM}$ expansion to subleading order for instantons $k=2,\dots 12$. The result matches the conjecture in \eqref{dlogZdm2More} for $\partial_m^2 \log Z \big\vert_{m=0}$ at finite $g_\text{YM}$ in terms of Eisenstein series, which generalizes the match found for the perturbative terms and the $k=1$ instanton in the main text. By explicitly performing the sums and products in $I_{p\times q}$ \eqref{eq:pq-tab} for many small values of $N$, we find that $I_{p\times q}$ can be expanded for small $a_i$ as
 \es{expansionApp}{
  I_{p\times q}(N,a_{ij}) &=  I_{p\times q} ^{(0)}(N) + I_{p\times q} ^{(2)}(N) C_2(a_{ij})+ \cdots \,,
 }
 where recall that $C_2$ is defined in \eqref{invars}. When $p=q$, which includes the one-instanton case \eqref{expansion}, we found closed form expressions for $I_{p\times q} ^{(0)}(N) $ and $I_{p\times q} ^{(2)}(N) $, but for $p\neq q$ we could only find recursion relations in $N$. In either case, these formulae can be expanded explicitly at large $N$. For instance, for $k=2$ we find the recursion relations:
\es{k2}{
&\left(-1600 N^2-7310 N-8256\right) I_{2\times 1}^{(0)}(N+2)+\left(1440 N^2+5859 N+5958\right) I_{2\times 1}^{(0)}(N+3)\\
&+(N+2) (160 N+491) I_{2\times 1}^{(0)}(N+1)=0\,,\qquad I_{2\times 1}^{(0)}(2)=-\frac{134}{27}\,,\qquad I_{2\times 1}^{(0)}(3)=-\frac{517}{81}, \\
}
and
\es{k222}{
&\left(-320 N^2-734 N-456\right)  I_{2\times 1}^{(2)}(N+2)+\left(288 N^2+1323 N+684\right) I_{2\times 1}^{(2)}(N+3)\\
&+N (32 N+51) I_{2\times 1}^{(2)}(N+1)=0\,,\qquad I_{2\times 1}^{(2)}(2)=-\frac{10}{243}\,,\qquad I_{2\times 1}^{(2)}(3)=-\frac{20}{729},\\
}
which we can solve in a large-$N$ expansion to get
\es{k22}{
&I_{2\times 1}^{(0)}=-5 \sqrt{\frac{2}{\pi }} \sqrt{N}+\frac{17 \sqrt{\frac{1}{N}}}{8 \sqrt{2 \pi }}+\frac{325 \left(\frac{1}{N}\right)^{3/2}}{1024 \sqrt{2 \pi }}+\frac{2155 \left(\frac{1}{N}\right)^{5/2}}{8192 \sqrt{2 \pi }}+\frac{1543605 \left(\frac{1}{N}\right)^{7/2}}{4194304 \sqrt{2 \pi }}+O(N^{-\frac92})\,,\\
&I_{2\times 1}^{(2)}=-\frac{15}{8 \sqrt{2 \pi } N^{5/2}}+\frac{255}{128 \sqrt{2 \pi } N^{7/2}}-\frac{11025}{16384 \sqrt{2 \pi } N^{9/2}}+\frac{478485}{131072 \sqrt{2
   \pi } N^{11/2}}+O(N^{-\frac{13}{2}})\,.
}
The cases $k=3,\dots,12$ are increasingly more complicated so we put them in an attached Mathematica file.

We can then use the expectation value $\langle C_2\rangle$ in \eqref{Expectations} with $\lambda=g_\text{YM}^2N$ to compute $\langle I_{p\times q}\rangle$ to $O(N^{-\frac92})$. The result in each case is consistent with the small $g_\text{YM}$ expansion to $O(g_\text{YM}^2)$ of
\es{ExpectationFinal}{
  &\left \langle  \partial_m^2 Z_\text{inst}^{p\times q} (m, a_{ij}) \right \rangle \bigg|_{m=0} =\frac{e^{ \frac{8pq \pi^2 }{g_\text{YM}^2} }}{1+\delta_{p,q}}\Big[ - \sqrt{N} \frac{16 K_1 (\frac{8pq \pi^2 }{ g_\text{YM}^2})}{g_\text{YM}}\left({\frac{p}{q}+\frac{q}{p}}\right)  +  \frac{2 K_2 (\frac{8pq \pi^2 }{ g_\text{YM}^2})}{g_\text{YM}\sqrt{N}}\left({\frac{p^2}{q^2}+\frac{q^2}{p^2}}\right)\\
 &+\frac{1}{32g_\text{YM}N^{\frac32}}\left[  -13 K_1 \left(\frac{8pq \pi^2 }{ g_\text{YM}^2}\right)\left({\frac{p}{q}+\frac{q}{p}}\right)+9 K_3 \left(\frac{8pq \pi^2 }{ g_\text{YM}^2}\right)\left({\frac{p^3}{q^3}+\frac{q^3}{p^3}}\right) \right] \\
 &+\frac{1}{128g_\text{YM}N^{\frac52}}\left[  -25  K_2 \left(\frac{8pq \pi^2 }{ g_\text{YM}^2}\right)\left({\frac{p^2}{q^2}+\frac{q^2}{p^2}}\right)+15  K_4 \left(\frac{8pq \pi^2 }{ g_\text{YM}^2}\right)\left({\frac{p^4}{q^4}+\frac{q^4}{p^4}}\right) \right] \\
 & +\frac{1}{g_\text{YM}N^{\frac72}}\left[\frac{1533 K_1\left(\frac{8pq \pi ^2}{g_\text{YM}^2}\right)}{16384 }\left({\frac{p}{q}+\frac{q}{p}}\right)-\frac{5355 K_3\left(\frac{8pq \pi ^2}{g_\text{YM}^2}\right)}{32768 }\left({\frac{p^3}{q^3}+\frac{q^3}{p^3}}\right)+\frac{2625
   K_5\left(\frac{8pq \pi ^2}{g_\text{YM}^2}\right)}{32768 }\left({\frac{p^5}{q^5}+\frac{q^5}{p^5}}\right)\right]\\
   &+ O(N^{-\frac92})\Big] \,,\\
 } 
which describes the $p\times q$ instanton terms in \eqref{dlogZdm2More}. This is a very nontrivial check of the conjectured finite $g_\text{YM}$ expression for $\partial_m^2 \log Z \big\vert_{m=0}$.

\bibliographystyle{ssg}
\bibliography{instanton}

\end{document}